%% file: main.tex
\documentclass[runningheads]{fundam}
\usepackage{graphicx}
\input{Package_and_commands/packages.tex}
\input{Package_and_commands/commands.tex}

\begin{document}
%
\title{Symbolic Model Checking using Intervals of Vectors }
%
%
\author{Damien Morard\orcidlink{0000-0001-6887-6357} \\
Computer Science Department, Faculty of Science,\\
University of Geneva, Switzerland\\
damien@morard.me \\
\and
Didier Buchs \\
Computer Science Department, Faculty of Science,\\
University of Geneva, Switzerland\\
didier.buchs@unige.ch
}

%
\maketitle              
\runninghead{D. Morard, D. Buchs}{Symbolic Model Checking using Intervals of Vectors}

\begin{abstract}
Model checking is a powerful technique for software verification. However, the approach notably suffers from the infamous state space explosion problem.
To tackle this, in this paper, we introduce a novel symbolic method for encoding Petri net markings.
It is based on the use of generalised intervals on vectors, as opposed to existing methods based on vectors of intervals such as Interval Decision Diagrams.
We develop a formalisation of these intervals, show that they possess homomorphic operations for model checking CTL on Petri nets, and define a canonical form that provides good performance characteristics. 
Our structure facilitates the symbolic evaluation of CTL formulas in the realm of global model checking, which aims to identify every state that satisfies a formula. 
Tests on examples of the model checking contest (MCC 2022) show that our approach yields promising results.
To achieve this, we implement efficient computations based on saturation and clustering principles derived from other symbolic model checking techniques.
\end{abstract}

\section{Introduction}

In the ever-evolving landscape of computer science and software engineering, the correctness and reliability of software systems remain paramount concerns.
In response to these challenges, the discipline of model checking has emerged as a powerful tool in the arsenal of software verification and validation.

Model checking~\cite{clarke:1997:model} is a formal verification technique used to ensure that a system adheres to its specifications and requirements.
However, although model checking is a trustworthy and robust method, it is confronted with a variety of challenges.
One of the most significant ones is the state space explosion problem~\cite{clarke:2011:model,mcmillan:1993:symbolic}.
This problem arises when dealing with complex software or hardware systems, where the number of possible states and transitions within a model grows exponentially with its size.
The task becomes more challenging when the objective is to identify all states that adhere to a given property, as opposed to merely determining the validity of the property within a specific configuration.
Moreover, our objective is to address the issue of \textit{global model checking}, namely, to identify all states that fulfill a CTL (Computation Tree Logic)~\cite{clarke:1981:design} formula.

While various techniques have been developed to mitigate the state space explosion problem, such as Decision diagrams~\cite{akers:1978:binary,hamez:2008:hierarchical,hostettler:2011:high}, abstractions~\cite{clarke:1994:model}, partial order reduction~\cite{peled:1994:combining}, and many others, it has become evident that this challenge will persist as models become more sophisticated.
Therefore, the quest for innovative methods to address this issue remains a compelling area of research.
As always, each new approach inevitably brings its own set of advantages and disadvantages compared to existing techniques.

This paper focusses on a symbolic method inspired by~\cite{racloz:1994:properties} that proposes to encode sets of markings in Petri net models with vectors that act as boundaries on the markings, similar to intervals.
As markings are essentially vectors of values, and fireability of transitions are lower constraints and their negation are upper constraints, intervals seem to be a suitable choice for representing the satisfiable states of a CTL formula.
However, expressing intervals on vectors is complicated, since vectors are not always comparable.
We introduce such notions and demonstrate that, for non-total orders, we require additional information such as a set of lower and upper bounds, rather than a single bound as in intervals of natural numbers.
Merely efficiently representing sets of markings is insufficient; we therefore also define homomorphic operations (w.r.t. union) on the symbolic structure to directly perform model checking on it.

Thanks to this encoded structure and operations, we can efficiently maintain the same structure and remain in the symbolic domain, without the need to decode it to perform computation over the encoded values.
Furthermore, we introduce a canonical form incorporating optimisations to make it efficient to use in the context of CTL model checking.
We tested our approach on examples from the model checking contest~\cite{mcc:2022}, yielding very promising results.

The paper is organised as follows: firstly, in \Cref{section:informal-presentation}, we provide an overview of our project. 
Next, in \Cref{section:svs}, we present the complete formalism of our symbolic structure.
In \Cref{section:from-svs-to-pn}, we demonstrate the connection between this symbolic structure and Petri nets to apply CTL model checking.
Finally, \Cref{section:benchmark} presents our results using this technique. This paper is an extended version of the paper published in Petri Nets 2024 \cite{DBLP:conf/apn/MorardDB24}.

\section{Informal presentation}
\label{section:informal-presentation}
In the context of Petri nets, as illustrated in~\Cref{figure:background:petri-net-first-visualisation}, the states of a model are represented using \textit{markings}.
A marking can be viewed as a vector where each place is associated with a natural number that indicates the number of resources.
In our example, $(3,1)$ is the current marking of the Petri net.

Model checking Petri nets consists in being able to answer queries such as "How many resources are required at each place to make $t_0$ fireable ?".
Two perspectives on the problem are possible from this question. The first seeks all states that satisfy the query, namely \textit{global model checking}~\cite{schuele:2004:global}.
The second checks whether a specific state satisfies the query, corresponding to \textit{local model checking}.
Obviously, because it computes the entire state space for a query, global model checking also provides an answer for any given state, making it more general than the local approach.
In particular, such a space is often infinite, meaning that a simple explicit enumeration of all states is impossible, as would be the case for the property of making $t_0$ fireable. 

Our approach uses a new structure called \textit{symbolic vector}, described by a couple $(a,b)$, where $a$ and $b$ represent sets of markings.
To belong to $(a,b)$, a given marking must include (given a partial order relation on vectors) all the markings in $a$ and none of the markings in $b$.
This is a generalisation of the interval concept to vectors.
Indeed, membership in a natural number interval $[a, b]$ can be expressed as $\forall x \in \mathbb{N}, a \leq x \leq b$, which can equivalently be formulated as $\forall x \in \mathbb{N}, a \leq x \land b \not< x$. 
The latter definition aligns with the conceptual framework of symbolic vectors, albeit in this case, the focus is on vectors rather than natural numbers.
The latter definition shares the same idea as symbolic vector, excepts that we work on vectors instead of naturals.
In the example in~\Cref{figure:background:petri-net-first-visualisation}, the symbolic vector encoding all solutions of $t_0$ being fireable is $(\set{(2,0)}, \varnothing)$.
This couple includes an infinite number of markings, such as $(2,0)$ or $(3,4)$ and so on.
Let us assume that we do not want to include markings greater than or equal to $(8,9)$; we would get the symbolic vector $(\set{(2,0)}, \set{(8,9)})$. 
In this refined version, valid markings could be $(7,7)$ or $(10,8)$, whereas invalid markings could be $(1,0)$, $(8,9)$ or $(10,9)$.
Therefore, based on an inclusion and exclusion set, we can encapsulate sets of markings.
In addition to this structure, we leverage homomorphic operations to perform computations directly on the symbolic structure, rather than on each element of the set individually.
For example, the action of adding $1$ to $p_0$ for all markings in $(\set{(2,0)}, \varnothing)$ can be done in a single step, leading to $(\set{(3,0)}, \varnothing)$, rather than having to iterate over each element in the set to increment it.

\begin{figure}[ht!]
\centering
\begin{tikzpicture}[>=stealth']
\node[place, label={135:$p_0$}] at (-2,0) (p0) {3};
\node[place, label={45:$p_1$}] at (2, 0) (p1) {$1$};
\node[transition, label={90:$t_0$}] at (0, 0) (t0) {};

\draw[->] (p0) to node[above] {$2$} (t0);
\draw[->] (t0) to node[above] {$1$} (p1);
\end{tikzpicture}

\caption{Petri net with two places and one transition}
\label{figure:background:petri-net-first-visualisation}
\end{figure}

Although symbolic vectors are a first step in encoding sets of markings, they are not sufficient to represent all possible sets, similar to intervals of natural numbers that cannot encode all sets of natural numbers.
To address this limitation, a set of intervals is required. 
For instance, the set of naturals $\set{1,2,3,7,8}$ is encoded as $\set{[1,3], [7,8]}$, necessitating two intervals.
However, constructing a symbolic version for sets poses a greater challenge.
Indeed, a symbolic representation is often linked to the task of finding a canonical form, as failing to do so may result in redundancy and other issues that diminish the efficiency of the approach.
For example, the set $\set{[1,4], [2,5], [3,6]}$ could be reduced to $\set{[1,6]}$ in natural numbers.
Due to the non-explicit representation, the uniqueness set rule is not sufficient to preserve the uniqueness of the symbolic set.
In fact, $[1,4] \neq [2,5]$ and $[1,4] \cap [2,5] \neq \varnothing$.
Each interval in this example is distinct, but all share values with the others.
The standard approach to constructing a canonical form in this context is to merge intervals whenever possible.

Unlike sets of intervals, creating a symbolic representation for sets of markings poses a more substantial challenge, primarily due to the non-strict partial order on its elements. 
Although all intervals can be compared to each other, the same is not true for markings.
Besides, a symbolic vector itself is an intricate structure that also needs to be canonised.
Consequently, additional constraints are necessary to manage canonicity of a \textit{symbolic vector set}.

By combining a symbolic encoding for sets of markings and encoding the evaluation process to be applied to symbolic vector sets, we are able to compute CTL formulas for global model checking, i.e. to determine all the markings of a net that satisfy a formula. 
It should be noted that the necessary operations for doing CTL model checking are mainly set operations and a \textit{pre} operation.
All of them are defined as homomorphisms on symbolic vector sets.

For the sake of simplicity, we only present the most important proofs in the remainder of this article.
Note that this work is a synthesis of the thesis in~\cite{morard:2024:symbolic} and all omitted proofs can be found in it.

\section{Symbolic structure for sets of vectors}
\label{section:svs}
This section formalises the structures known as \textit{symbolic vector} and \textit{symbolic vector set}.
Symbolic vectors serve as the initial encoding layer for Petri net markings, whereas sets of symbolic vectors constitute the set that incorporates them.
We present all of their respective definitions, properties, and canonical form.

\subsection{Symbolic vectors}

\begin{definition}[Non-strict partial order]
Let $a,b \in Q$ be two elements.
A \textit{non-strict partial order} is a homogeneous relation $\subseteq_f$ on a set $Q$, such that it verifies the following properties; reflexivity: $a \subseteq_f a$; antisymmetry: if $a \subseteq_f b \land b \subseteq_f a $ then $a = b$; transitivity: if $a \subseteq_f b \land b \subseteq_f c$ then $a \subseteq_f c$.
\end{definition}
Notice that two values can potentially not be compared, i.e. neither $a \subseteq_f b$ nor $b \subseteq_f a$ are in the relation.

\begin{definition}[Inclusion relation on vectors]
    Let $q_a, q_b \in Q_n$ be two elements such that $Q_n$ is the set of vectors of size $n \in \Nat$ over naturals $\Nat$  as elements. \footnote{In all generality it could have been any ordered monoid}
    The non-strict partial order on two vectors of the same size, namely the \textit{inclusion of vectors}, is a relation defined as $q_a \subseteq_f q_b \Leftrightarrow \forall p \in P, q_a(p) \leq q_b(p)$.
\end{definition}

We use the notation $\forall p \in P, q(p)$ where $q \in Q_n$ to iterate over the elements of our tuple.
This parallels a function structure, where each tuple location serves as the domain and the associated value functions as its codomain.
From the previous definition, the derivations of the definitions for $\not\subset_f, =_f, \neq_f$ can be inferred.
In addition, for the remainder of the article, we assume the same definition of $Q_n$.

\begin{definition}[Symbolic vector]
\label{definition:predicate-structure:belongs-to}
Let $Q_n$ be a set of vectors such that there exists a non-strict partial order $\subseteq_f$ on its elements.
A \textit{symbolic vector} is defined as: $(a,b) \in \SV, ~\text{where}~a,b \in \powerset{Q_n}$.
A vector that belongs to it, noted $\in_{sv}: Q_n \times \SV$, is defined by the relation $(q \in_{sv} (a,b)) \Leftrightarrow \forall q_a \in a, q_a \subseteq_f q \land \forall q_b \in b, (q_b \not\subset_f q \land q_b \neq_f q)$. We assume that $\not\in_{sv}$ is its negation.
\end{definition}

\begin{example}
We illustrate this with a few examples of vectors that either belong or do not belong to the corresponding symbolic vectors.
\begin{align*}
(0,3) &\in_{sv} (\{(0,1)\},\{(1,4)\})\\
(0,3) &\in_{sv} (\{(0,1),(0,2)\},\{(1,4),(0,5)\})\\
(0,3) &\not\in_{sv} (\{(1,1)\},\{(1,4)\})
\end{align*}
\end{example}

A symbolic vector relies on set theory and the definition of a non-strict partial order relation for elements inside it.
For a vector $q$ to be accepted by such a structure, all elements of $a$ must be included in $q$ and none of  the elements of $b$ must be included in it.
Note that the whole relation may be seen as a sequence of conjunctions ($\land$) where each predicate must be true.

\begin{remark}
Note that the theory of symbolic vectors could be generalised to any structure with a non-strict partial order.
However, for simplicity and a clear focus on Petri nets, we will describe it directly for vectors, the structure used to represent markings. 
Moreover, one advantage of vectors of totally ordered values is that there exists a lexicographic order on them that can be used as a total order.
Thus, it can help construct the canonical form.
\end{remark}

\begin{definition}[Underlying set of vectors]
\label{definition:predicate-structure:uf}
Let $sv \in \SV$ be a symbolic vector.
The \textit{set of underlying vectors} of a symbolic vector is a function $uf: \SV \rightarrow \powerset{Q_n}$ such that $uf(sv) = \set{q \in Q_n~|~ q \in_{sv} sv}$.
\end{definition}

\begin{figure}[ht!]
\centering
\scalebox{0.8}{
    \begin{minipage}[b]{0.48\textwidth}
        \centering
        \begin{tikzpicture}
            \draw[->] (0,0) -- (5,0) node[right] {};
            \draw[->] (0,0) -- (0,4) node[above] {};
        
            \foreach \x in {0,1,...,5}
                \draw (\x,0.1) -- (\x,-0.1) node[below] {$\x$};
        
            \foreach \y in {0,1,...,4}
                \draw (0.1,\y) -- (-0.1,\y) node[left] {$\y$};
        
            \draw[draw=none, pattern=dots] (1,0) rectangle (5,1);
            \draw[draw=none, pattern=dots] (0,1) rectangle (1,4);
            \draw[draw=none, pattern=north west lines] (1,1) rectangle (5,4);
        
            \draw[dashed] (1,0) -- (1,4) node[above] at (1,4) {};
            \draw[dashed] (0,1) -- (5,1) node[above] at (5,1) {};
        
        \end{tikzpicture}
        \captionsetup{font=footnotesize}
        \captionof{figure}{Visualisation of the symbolic vector $(\set{(1,0), (0,1)}, \varnothing)$.}
        \label{figure:graph-visualisation1}
    \end{minipage}
}
\scalebox{0.8}{
    \begin{minipage}[b]{0.48\textwidth}
        \centering
        \begin{tikzpicture}
            \draw[->] (0,0) -- (5,0) node[right] {};
            \draw[->] (0,0) -- (0,4) node[above] {};
        
            \foreach \x in {0,1,...,5}
                \draw (\x,0.1) -- (\x,-0.1) node[below] {$\x$};
        
            \foreach \y in {0,1,...,4}
                \draw (0.1,\y) -- (-0.1,\y) node[left] {$\y$};
        
            \draw[draw=none, pattern=dots] (1,0) rectangle (5,1);
            \draw[draw=none, pattern=dots] (0,1) rectangle (1,4);
            \draw[draw=none, pattern=north west lines] (1,1) rectangle (2,4);
            \draw[draw=none, pattern=north west lines] (2,1) rectangle (4,3);
            \draw[draw=none, pattern=checkerboard] (2,3) rectangle (5,4);
            \draw[draw=none, pattern=checkerboard] (4,1) rectangle (5,3);
        
            \draw[dashed] (1,0) -- (1,4) node[above] at (1,4) {};
            \draw[dashed] (0,1) -- (5,1) node[above] at (5,1) {};
        
        \end{tikzpicture}
        \captionsetup{font=footnotesize}
        \captionof{figure}{Visualisation of the symbolic vector $(\set{(1,0), (0,1)}, \set{(2,3), (4,1)})$.}
        \label{figure:graph-visualisation2}
    \end{minipage}
}
\end{figure}

In~\Cref{figure:graph-visualisation1}, we present an initial example of visualisation \footnote{the upper bound on both axis is artificially limited for the sake of presentability} for the underlying set of vectors from $(\set{(1,0), (0,1)}, \varnothing)$.
The coordinate of a value in this grid is represented as a two-dimensional vector.
Vectors in the dot pattern satisfy only one of the conditions, i.e., $\forall q \in Q_2, ((1,0) \subseteq q \land (0,1) \not\subset q \lor ((0,1) \subseteq q \land (1,0) \not\subset q)$.
In addition, those in the hatch pattern satisfy both simultaneously, i.e. $\forall q \in Q_2, (1,0) \subseteq q \land (0,1) \subseteq q$.
This illustrates that the left vectors must all be included, denoted by the "\textit{and}" operator. 
In contrast, when adding multiple constraints, as shown in~\Cref{figure:graph-visualisation2} with the chequerboard pattern, the "\textit{or}" operator comes into play, making a solution invalid if it includes $(2,3)$, $(4,1)$, or both simultaneously.
It could be formally rewritten as $\forall q \in Q_2, (1,0) \subseteq q \land (0,1) \subseteq q \land (2,3) \not\subset q \land (4,1) \not\subset q$.
Therefore, the valid vectors are those belonging to the hatch pattern. 
Note that, while we bounded the domain size for the graphical representation, the hatch pattern symbolises an infinite number of solutions.
Moreover, while this visualisation is effective for a vector of size 2, it needs adaptation for varying vector lengths. 
Nevertheless, the fundamental concept remains consistent across different vector sizes.

\begin{definition}[Intersection of two symbolic vectors]
\label{def:predicate-structure:intersection-ps}
Let $sv = (a,b), sv' = (c,d) \in \SV$ be two symbolic vectors.
The \textit{intersection} between two symbolic vectors, noted $\cap_{sv}: \SV \times \SV \rightarrow \SV$, is defined as: $sv \cap_{sv} sv' = (a \cup c, b \cup d)$.
\end{definition}

The intersection takes the conditions of each symbolic vector to combine them.
Thus, to belong to their intersection, a new vector must belong to $a$ and $c$, but not to $b$ and $d$.
Furthermore, the intersection is the only operation on symbolic vectors that gives a symbolic vector as an output by combining two of them.
Similarly to intervals, the difference or union may return a set of symbolic intervals instead of a single symbolic interval.

\begin{lemma}
\label{lemma:predicate-structure:intersection-homomorphism-sv}
Let $sv, sv' \in \SV$, then $uf(sv \cap_{sv} sv') = uf(sv) \cap uf(sv')$.
\end{lemma}

Proof of~\Cref{lemma:predicate-structure:intersection-homomorphism-sv}:
We want to prove the following lemma.
Let $sv, sv' \in \SV$, then $uf(sv \cap_{sv} sv') = uf(sv) \cap uf(sv')$.
\begin{proof}
    \begin{align*}
        uf&(sv \cap_{sv} sv') = uf((a,b) \cap (c,d)) & (\text{\Cref{def:predicate-structure:intersection-ps}})\\
        &= \set{q ~|~ q \in_{sv} (a \cup c, b \cup d))} & (\text{\Cref{definition:predicate-structure:uf}})\\
        &= \set{q ~|~ \forall q' \in (a \cup c), q' \subseteq_f q \land \forall q'' \in (b \cup d), (q'' \not\subset_f q \land q'' \neq_f q) } & (\text{\Cref{definition:predicate-structure:belongs-to}}) \\
        &= \set{q ~|~ \forall q_a \in a, q_a \subseteq_f q \land \forall q_b \in b, (q_b \not\subset_f q \land q_b \neq_f q) &\text{(Set rules)} \\
        & ~~~~\land \forall q_c \in c, q_c \subseteq_f q \land \forall q_d \in d, (q_d \not\subset_f q \land q_d \neq_f q)} \\
        &=  \set{q ~|~ \forall q_a \in a, q_a \subseteq_f q \land \forall q_b \in b, (q_b \not\subset_f q \land q_b \neq_f q)} &\text{(Set rules)} \\
        & ~~~~\cap \set{q' ~|~ \forall q_c \in c, q_c \subseteq_f q' \land \forall q_d \in d, (q_d \not\subset_f q' \land q_d \neq_f q')} \\
        &= \set{q ~|~ q \in_{sv} (a,b)} \cap \set{q' ~|~ q' \in_{sv} (c,d)} & (\text{\Cref{definition:predicate-structure:belongs-to}}) \\
        &= uf(a,b) \cap uf(c,d) & (\text{\Cref{def:predicate-structure:intersection-ps}}) \\
        &= uf(sv) \cap uf(sv')
    \end{align*}
\end{proof}

In the following, we assume that $f_\varepsilon \in Q_n$ is the \textit{zero vector} such that all its values are set to 0.

\begin{example}
We illustrate this with a few examples of vectors to represent specific sets.
\begin{align*}
\emptyset & = uf(\{(1,1)\},\{(1,1)\}) = uf(\{(3,1)\},\{(1,1)\}= uf(\{f_\epsilon\},\{f_\epsilon\})\\
Q_n & = uf(\emptyset,\emptyset)
\end{align*}
\end{example}

\begin{definition}[Join]
\label{definition:predicate-structure:convergent-function}
Let us assume $p \in \powerset{Q}$ is a set of vectors.
The \textit{join} function, which finds the \textit{least upper bound} of a set of vectors, written as $join: \powerset{Q_n} \rightarrow \powerset{Q_n}$, is defined as $join(\varnothing) = \set{f_\varepsilon}$ and
$join(p) = \set{q_{join}}, p \neq \varnothing,~\text{such that}~\forall q \in p, (q \subseteq_f q_{join} \land \not\exists q' \in \powerset{Q}, q' \subseteq_f q_{join} \land q_{join} \neq_f q' \land q \subseteq q')$.
\end{definition}

The $join$ operation has its origins in lattice theory and has been adapted for direct application on vectors.
For example, $join(\set{(1,1), (9,9)}) = \set{(9,9)}$ or $join(\set{(1,3), (7,2), (4,4)}) = \set{(7,4)}$.
Note that $join$ returns a singleton.
Additionally, when the input is an empty set, it returns a vector filled with zeros, representing the global least upper bound, ensuring its inclusion in all vectors.

\begin{lemma}
\label{lemma:predicate-structure:reduction-conv}
Let $sv = (a,b) \in \SV$.
Then, $uf(a,b) = uf(join(a),b)$.
\end{lemma}

Thanks to \Cref{lemma:predicate-structure:reduction-conv}, we can treat $a$ as a singleton, which is an important simplification.
However, the same idea cannot be applied to $b$.
Indeed, the reduction of the set $a$ capitalises on the '\textit{and}' relation between vectors, while $b$ cannot achieve the same simplification due to the '\textit{or}' relation.

\begin{definition}[Empty symbolic vector]
\label{definition:predicate-structure:empty-uf}
Let $sv = (a,b) \in \SV$ be a symbolic vector and $\set{q_a} = join(a)$\footnote{If $a = \varnothing$, then $join(a) = \set{f_\varepsilon}$, where $q_a$ is included in all markings of $b$}.
An empty symbolic vector, noted $\varnothing_{sv}$, is defined as $sv \in \varnothing_{sv} \Leftrightarrow \exists q_b \in b, q_b \subseteq_f q_a$.
Furthermore, when expressing $sv \not\in \varnothing_{sv}$, we assume its negation.
\end{definition}

In fact, addressing emptiness in symbolic vectors is a crucial challenge. 
Given their infinite number, there is considerable potential for redundancy, necessitating the ability to identify them for the development of a future canonical form.
For example, $uf(\set{(3,3)}, \set{(1,1)}) = uf(\set{(0,0)}, \set{(0,0)}) = \varnothing$.

\begin{definition}[Canonicity of symbolic vectors]
Let $sv = (a,b) \in \SV$ be a symbolic vector.
$sv$ is \textit{canonical} if and only if:
\begin{enumerate}
    \item $a = \set{q_a}$ is a singleton.
    \item $sv \not\in \varnothing_{sv} \lor sv = (\set{f_\varepsilon}, \set{f_\varepsilon})$.
    \item $\forall q_b, q_b' \in b, q_b \neq q_b', q_b$ and $q_b'$ are not \textit{comparable}.
    \item $\forall q_b \in b, \forall p \in P, q_a(p) \leq q_b(p)$.
\end{enumerate}
\end{definition}

To ensure the canonicity of a symbolic vector, several conditions must be satisfied, each contributing to its unique representation.
Firstly, we rely on the join function described in \Cref{definition:predicate-structure:convergent-function} as the initial condition, which ensures that $a$ is always expressible as a singleton.
Secondly, addressing the issue of multiple possible expressions for an empty symbolic vector, a deliberate decision has been made to maintain a single representation, namely $(\set{f_\varepsilon}, \set{f_\varepsilon})$.
Thirdly, the vectors in the set $b$ must not be encapsulated within a vector in the same set, since the most inclusive vector already encapsulates all constraints.
For example, $uf(\set{(1,1)}, \set{(4,4), (5,5)}) = uf(\set{(1,1)}, \set{(4,4)})$, because $(4,4) \subseteq_f (5,5)$.
Lastly, the fourth condition ensures that each value of $q_a$ serves as a minimum bound for each value of each vector in $b$, as illustrated in the following exemple.
\begin{example}
    $uf(\set{(2,0,5)}, \set{(4,0,2)}) = uf(\set{(2,0,5)}, \set{(4,0,5)})$.
    Both vectors have the same underlying representation.
    Although $(4,0,2) \neq_f (4,0,5)$, the minimum vector is $(2,0,5)$, which forces each component to be greater than or equal to it to be accepted.
    Therefore, as long as the value of the third component is less than or equal to 5 for the right part, the underlying set remains equivalent.
    Furthermore, by imposing the minimum value of $q_a$ in $b$, all vectors in $b$ are inherently included in $a$, i.e. $\forall q_b \in b, q_a \subseteq_f q_b$.   
\end{example}

\begin{definition}[Symbolic equality]
\label{theorem:canonical-ps:canonicity-ps}
Let $sv, sv' \in \SV$ be two symbolic vectors.
Then, $uf(sv) = uf(sv') \Leftrightarrow sv =_{sv} sv'$.
\end{definition}


\begin{remark}
In the following, we assume a function $can_{sv} : \SV \rightarrow \SV$ that returns the canonical form of a symbolic vector.
A full description of this function would require introducing the complete collection of operations used to compute the canonical form, which is provided in~\cite{morard:2024:symbolic}.
Note that we will later introduce a distinct canonicalisation function, denoted by $can_{svs}$, which operates on sets rather than on symbolic vectors.
\end{remark}

\begin{lemma}
\label{lemma:canonical-ps:basic-equality}
Let $sv \in \SV$ be a symbolic vector.
Then, $can_{sv}(sv) =_{sv} sv$.
\end{lemma}

\begin{example}
$can_{sv}(\set{(1,2), (2,1)}, \set{(0,5), (0,7)}) =_{sv} (\set{(2,2)}, \set{(2,5)})$.

The canonical form is obtained through the following three main steps:
\begin{enumerate}
\item Apply the $join$ operation: $join(\set{(1,2),(2,1)}) = \set{(2,2)}$.
\item Remove the upper bound $(0,7)$, because the vector $(0,5)$ includes all values greater than $(0,7)$.
\item Modify the remaining upper bound $(0,5)$ according to the lower bound $\set{(2,2)}$, yielding $(2,5)$.
\end{enumerate}
\end{example}

\subsection{Symbolic vector sets}

As mentioned before when using intervals, union is not an internal operation of intervals. The same occurs for symbolic vectors. We then need the following definitions and properties that are based on the previous definition of symbolic vectors and can be seen as extensions that work on sets.

\begin{definition}[Symbolic vector sets]
A symbolic vector set, belonging to $\SVS$, is a set containing exclusively symbolic vectors.
A vector that belongs to it, noted $\in_{svs}: Q_n \times \SVS$, is defined by the relation $(q \in_{svs} svs) \Leftrightarrow \exists sv \in svs, q \in_{sv} sv$.
We assume $\not\in_{svs}$ as its negation.
\end{definition}

The equivalent class of the \textit{empty symbolic vector set} is: $svs \in \varnothing_{svs} \Leftrightarrow \forall sv \in svs, sv \in \varnothing_{sv}$, where $\varnothing_{svs}$ is the set of all empty symbolic vector sets.
The empty set $\varnothing$ also belongs to $\varnothing_{svs}$.

\begin{definition}[Underlying vectors of set]
\label{def:predicate-structure:underlying-vectors-of-set}
The \textit{set of underlying vectors} of a symbolic vector set is a function, noted $\usf: \SVS \rightarrow \powerset{Q_n}$, such that:
\begin{align*}
    \usf(svs) = \bigcup_{sv~\in~svs}(uf(sv))
\end{align*}
\end{definition}

Similarly to $uf$, we unfold the result of each symbolic vector to obtain the final vector set.

\begin{definition}[Union \& intersection]
\label{def:predicate-structure:union-intersection}
Let $svs, svs' \in \powerset{\SV}$ be two symbolic vector sets.
The \textit{union} and \textit{intersection} between two symbolic vector sets, noted $\cup_{svs}, \cap_{svs}: \SVS \times \SVS \rightarrow \SVS$, are defined as: $svs \cup_{svs} svs' = svs \cup svs'$ and $svs \cap_{svs} svs' = \bigcup_{sv \in svs}~\bigcup_{sv' \in svs'} sv~ \cap_{sv} sv'$.
\end{definition}

Union for symbolic vector sets is simply set union.
The intersection of two symbolic vector sets is the combination of the intersection between all the symbolic vectors of $sv$ and $sv'$.

\begin{lemma}[Commutativity]
\label{lemma:predicate-structure:commutativity-svs}
Let $svs, svs' \in \SVS$ be two symbolic vector sets and $\ast \in \set{\cup_{svs}, \cap_{svs}}$
Then, $svs \ast svs' = svs' \ast svs$.
\end{lemma}

Before introducing the difference between two symbolic vector sets, we formalise the difference of two symbolic vectors.

\begin{definition}[Difference of symbolic vectors]
\label{definition:predicate-structure:difference}
Let $sv = (a,b), sv' = (c,d) \in \SV$ be symbolic vectors.
The \textit{difference} between two symbolic vectors, noted $\setminus_{sv}: \SV \times \SV \rightarrow \SVS$, is defined as:
\begin{align*}
    sv \setminus_{sv} sv' = \{(a, b \cup c)\} \cup_{svs}~\bigcup_{q_d \in d} (a \cup c \cup \set{q_d}, b)
\end{align*}
\end{definition}

We remove all the values that are common to both symbolic vectors.
$c$ and $d$ are used as upper and lower bounds, respectively.
Note that in the construction of the new set we have $b \cup \set{q_c}$ and $a \cup c \cup \set{q_d}$ instead of $\set{q_c}$ and $\set{q_d}$.
The omission of $c$ in the construction of the latter could cause a problem when $q_c \not \leq q_d$.
This is crucial to preserve the initial condition in addition to the new bounds, preventing acceptance of values that were not part of the original conditions.

\begin{example}
    Let us illustrate with an example:
    \begin{align*}
        &(\set{(1,1)}, \set{(9,9)}) \setminus_{sv} (\set{(4,4)}, \set{(7,2), (5,6)}) \\
        =&~\set{(\set{(1,1)}, \set{(9,9)} \cup \set{(4,4)})} \cup_{svs} \set{(\set{(1,1)} \cup \set{(4,4)} \cup \set{(7,2)}, \set{(9,9)})} \\
        &\cup_{svs} \set{(\set{(1,1)} \cup \set{(4,4)} \cup \set{(5,6)}, \set{(9,9)})} \\
        =&~\set{(\set{(1,1)}, \set{(4,4)}), (\set{(7,4)}, \set{(9,9)}), (\set{(5,6)}, \set{(9,9)})}
    \end{align*}
    The last step involves simplifications, ensuring that each symbolic vector adheres to the canonical form.
    This primarily involves using $join$ for the left component of each symbolic vector.
    Then, we verify that there are no comparable vectors in $b$, as exemplified by $\set{(9,9)} \cup \set{(4,4)}$, reduced to $\set{(4,4)}$.
    Note that this operation may introduce redundancy in the resulting symbolic vector set.
    The vector $(7,6)$ belongs to $(\set{(7,4)}, \set{(9,9)})$ and $(\set{(5,6)}, \set{(9,9)})$.
\end{example}

\begin{lemma}
\label{lemma:predicate-structure:difference-ps-homomorphism}
Let $sv, sv' \in \SV$. 
Then, $\usf(sv \setminus_{sv} sv') = uf(sv) \setminus uf(sv')$.
\end{lemma}

\begin{definition}[Difference of symbolic vector sets]
\label{definition:predicate-structure:difference-svs}
Let $svs = \set{sv_1, \dots, sv_n}, svs' \in \SVS, n > 1$ and $sv, sv' \in \SV$.
The \textit{difference} between two symbolic vector sets, noted $\setminus_{svs}: \SVS \times \SVS \rightarrow \SVS$, is defined inductively as:
\begin{align*}
svs \setminus_{svs} \varnothing_{svs} &= svs,~\varnothing_{svs} \setminus_{svs} svs = \varnothing_{svs} \\
\set{sv} \setminus_{svs} \set{sv'} &= sv \setminus_{sv} sv' \\
\set{sv} \setminus_{svs} \set{sv_1, \dots, sv_n} &= (sv \setminus_{sv} sv_1) \setminus_{svs} \set{sv_2, \dots, sv_n} \\
\set{sv_1, \dots, sv_n} \setminus_{svs} svs' &= \set{sv_1} \setminus_{svs} svs' \cup_{svs} \set{sv_2, \dots, sv_n} \setminus_{svs} svs'
\end{align*}
\end{definition}

The difference is carried out recursively by taking each left symbolic vector and applying the difference with each right symbolic vector.
Moreover, since the result of $\setminus_{sv}$ is a new set, we need to invoke $\setminus_{svs}$.

\begin{definition}[Negation]
Let $svs \in \SVS$ be a symbolic vector set  and $sv \in \SV$ be a symbolic vector.
The respective \textit{negation} of both structures, noted $\neg_{svs}: \SVS \rightarrow \SVS$ and $\neg_{sv}: \SV \rightarrow \SVS$, are defined as: 
$\neg_{sv}(sv) = (\set{f_\varepsilon}, \varnothing) \setminus_{sv} sv~\text{and}~\neg_{svs}(svs) = \set{(\set{f_\varepsilon}, \varnothing)} \setminus_{svs} svs $
\end{definition}

It must be noted that the negation is defined as a complement to the whole set.
\begin{lemma}[Homomorphic operations]
Let $svs, svs' \in \SVS$ and $\ast \in \set{\cup, \cap, \setminus}$ be operations.
Then, $\\\usf(svs) \ast \usf(svs') = \usf(svs \ast_{svs} svs')$ and $\usf(\neg_{svs} svs) = \neg \usf(svs)$
\end{lemma}

\begin{definition}[Symbolic vector set relations]
Let $svs, svs' \in \SVS$ be symbolic vector sets and $q \in Q_n$ be a vector.
Each of the relations $\subseteq_{svs}, \not\subset_{svs}, =_{svs}, \neq_{svs}\;\subseteq \SVS \times \SVS$ and $\in_{svs}\;\subseteq \SV \times \SVS$ are defined as follows:
\begin{itemize}
    \item $svs \subseteq_{svs} svs' \Leftrightarrow svs \setminus_{svs} svs' \in \varnothing_{svs}$
    \item $svs \not\subset_{svs} svs' \Leftrightarrow svs \setminus_{svs} svs' \notin \varnothing_{svs}$ \footnote{For legacy reasons we used the symbol $\not\subset$ for the negation of $\subseteq$ which could be a bit awkward}
    \item $svs =_{svs} svs' \Leftrightarrow (svs \subseteq_{svs} svs') \land (svs' \subseteq_{svs} svs)$
    \item $svs \neq_{svs} svs' \Leftrightarrow (svs \not\subset_{svs} svs') \lor (svs' \not\subset_{svs} svs)$
    \item $q \in_{svs} svs \Leftrightarrow \exists sv \in svs, q \in_{sv} sv$
\end{itemize}
\end{definition}

The general set relations are expanded to directly work on sets of symbolic vectors.
Their constructions are based on the operations described previously.

\subsection{Total order, mergeability and shareability of symbolic vector sets}

Establishing a canonical form for symbolic vector sets requires a deep understanding of the interactions that may occur between symbolic vectors, in much the same way as one must understand the interactions between integer intervals in order to reduce them effectively.
Among the conditions involved is the need to prevent overlap between intervals, which corresponds to avoiding the inclusion of the same value in two distinct symbolic vector sets.
Additionally, another scenario may arise when two intervals share no values, but can be merged into a single one. 
For instance, the two intervals [1,4] and [5,7] in natural numbers can be combined into [1,7]. 
Therefore, symbolic vector sets must take into account these conditions.

\begin{figure}[ht!]
\centering
\scalebox{0.8}{
\begin{minipage}[b]{0.48\textwidth}
    \centering
    \begin{align*}
        &\usf(\set{(\set{(2,4)}, \varnothing), (\set{(4,2)}, \set{(4,4)})}) \\
        &~~~~~~~~~~~~~~~~~~~~~~~~~~=\\
        &\usf(\set{(\set{(2,4)}, \set{(4,4)}), (\set{(4,2)}, \varnothing)}) \\
    \end{align*}
    \captionof{figure}{Two symbolic vector sets with the same underlying set of vectors.}
    \label{figure:canonical-ps:shearbility11}
\end{minipage}
}
\scalebox{0.8}{
\begin{minipage}[b]{0.48\textwidth}
    \centering
        \begin{tikzpicture}[
            > = stealth, 
            shorten > = 1pt, 
            auto,
            node distance = 1cm, 
            draw = black,
            thick,
            fill = white,
            minimum size = 4mm
        ]
            \node[ellipse, draw] at (-1,0) (24) {(2,4)};
            \node[ellipse, draw] at (1,0) (42) {(4,2)};
            \node[ellipse, draw] at (0,1) (44) {$(4,4)$};
            \node[ellipse] at (0,2.25) (infinf) {$...$};
            
            \draw[] (24) to node {} (44);
            \draw[] (42) to node {} (44);
            \draw[] (44) to node {} (infinf);
        \end{tikzpicture}
    \captionof{figure}{Visualisation of the two symbolic vector sets.}
    \label{figure:canonical-ps:shearbility12}
\end{minipage}
}
\caption{
Illustration of the canonical issue encountered with the non-strict partial order on its vectors, where $(\set{(4,4)}, \varnothing)$ is potentially contained in both symbolic vectors.
}
\label{figure:canonical-ps:shareability1}
\end{figure}

However, addressing the aforementioned cases is insufficient to guarantee future canonicity due to the non-strict partial order of its elements. 
\Cref{figure:canonical-ps:shareability1} depicts a scenario in which two symbolic vectors do not share values and cannot be merged.
Despite this, two different sets exist, resulting in an equivalent final underlying set of vectors.
The issue arises because the value within the symbolic vector $(\set{(4,4)}, \varnothing)$ may be included in both symbolic vectors. 
The visualisation in~\Cref{figure:canonical-ps:shearbility12} offers an initial explanation of this phenomenon.
The nodes are arranged according to their vectors, with the lowest at the bottom and the largest at the top.
Vector (4,4) serves as a \textit{join} point where the two symbolic vectors share the above values.
Moreover, vector (4,4) can be moved over both excluding sets to transfer the constraint, leading to two different representations for the same underlying set.
Thus, it is imperative to establish a deterministic rule for selecting the symbolic vector that should encompass the constraint.
To address this, we introduce a total order on symbolic vectors.

\begin{definition}[Lexicographic ordering]
\label{def:canonical-ps:lexicographic-ordering}
Let $q = (x_1, \dots, x_n), q' = (y_1, \dots ,y_n) \in Q_n$. 
The \textit{lexicographic ordering}, noted $\leq_f: Q_n \times Q_n \rightarrow \Bool$, is defined as:
\begin{align*}
() \leq_f () &= \true \\
(x_1, \dots, x_n) \leq_f (y_1, \dots, y_n) &= \begin{cases}
    true &x_1 < y_1 \\
    false &y_1 < x_1 \\
    (x_2, \dots, x_n) \leq_f (y_2, \dots, y_n) &x_1 = y_1
\end{cases}
\end{align*}
\end{definition}

\begin{lemma}
    $\leq_f$ defines a total order on $Q_n$.
\end{lemma}

Using this new relation, we can determine the symbolic vector that should incorporate the constraint.
Referring to the previous illustration in \Cref{figure:canonical-ps:shareability1}, we opted to retain the constraint in the larger symbolic vector. 
Consequently, the selected form that we keep among these two alternatives is $\set{(\set{(2,4)}, \varnothing), (\set{(4,2)}, \set{(4,4)})}$.


\begin{definition}[Mergeable]
Let $sv = \set{(\set{q_a}, b)}, sv' = \set{(\set{q_c}, d)} \in \SV$ be canonical symbolic vectors.
The mergeable relation on two symbolic vectors, noted $mergeable_{sv} \subseteq \SV \times \SV$, is defined as:
\begin{align*}
    \bm{if}~&sv \in \varnothing_{sv} \lor sv' \in \varnothing_{sv}: ~mergeable_{sv}(sv, sv') = \true \\
    \bm{if}~&sv \not\in \varnothing_{sv} \land sv' \not\in \varnothing_{sv}& \\
    mergeable_{sv}(sv, sv') &=
    \begin{cases}
        \bm{if}&q_a \subseteq_f q_c \\
        &\begin{cases}
            \bm{if}& \forall q_b \in b, (q_c \subseteq_f q_b \lor~\exists q_d \in d, q_d \subseteq_f join(\set{q_b, q_c})) \\
            &~\true \\
            \bm{else}&\false \\
        \end{cases} \\
        \bm{if}&q_c \subseteq_f q_a~mergeable_{sv}(sv',sv) \\
        \bm{else}& \false
        \end{cases} 
\end{align*}
\end{definition}

The $mergeable_{sv}$ function determines whether two canonical symbolic vectors can be combined into a singleton.
The first essential condition is that $q_a$ and $q_c$ must be comparable.
If $q_a \subseteq_f q_c$, one of the two following conditions must be satisfied for each vector in $b$ to return true.
The intuition behind the condition $q_c \subseteq_f q_b$ can be seen by analogy with integer intervals such as $[1,5]$ and $[3,8]$, where $3 \leq 5$.
In contrast, the condition $\exists q_d \in d, q_d \subseteq_f join(\set{q_b, q_c})$ is intended to ensure that the upper bound resulting from the join operation is not exceeded by either of the original upper bounds; otherwise, some vectors would be accepted when it should not be.

Note that if at least one of the symbolic vectors belongs to $\varnothing_{sv}$, $mergeable_{sv}$ returns $\true$.

\begin{example}
Let us illustrate different examples of symbolic vectors that can be merged or not.
\begin{enumerate}
\item $mergeable((\set{(1,1)},\set{(7,7)}), (\set{(3,3)},\set{(5,5)})) = \true,\\$ resulting in $\set{(\set{(1,1)},\set{(7,7)}}).\\$ 
This is the easiest case where $(1,1) \subseteq_f (3,3)$ and $(3,3) \subseteq_f (7,7)$.
In this context, this case is similar to a simple comparison of intervals such as $[1,7)$ and $[3,5)$, where it seems clear that the interval $[3,5)$ is included in $[1,7)$.
\item $mergeable((\set{(1,1)},\set{(4,4)}), (\set{(6,6)},\set{(8,8)})) = \false. \\$ Although $(1,1) \subseteq_f (6,6)$, we have $(6,6) \not\subset_f (4,4)$ and $(8,8) \not\subset_f (6,6) = join(\set{(4,4), (6,6)})$.
There exists a gap between the two symbolic vectors.
This can also be interpreted in terms of integer intervals, such as $[1,4)$ and $[6,8)$, where the gap corresponds to the set of integers in $[4,6)$.
\item $mergeable((\set{(1,3)},\set{(2,3),(1,4)}), (\set{(1,4)},\set{(2,4), (1,5)})) = \true$, \\resulting in $\set{(\set{(1,3)},\set{(2,3), (1,5)}}).\\$
This is the scenario in which two symbolic vectors follow each other.
\item $mergeable((\set{(1,3)},\set{(2,3),(1,4)}), (\set{(1,4)},\set{(2,5), (1,5)})) = \false.\\$
This is a modified version of the latter to illustrate the scenario where the condition $\exists q_d \in d, q_d \subseteq_f join(\set{q_b, q_c})$ does not hold. The second symbolic vector has the function $(2,5)$ instead of $(2,4)$.
Fusion is no longer possible due to $(2,5) \not\subset_f (2,4)$, where $\set{(2,4)} = join(\set{(2,3),(1,4)})$.
If we had merged both symbolic vectors as $\set{(\set{(1,3)},\set{(2,5), (1,4)}})$, we would have accepted the vector $(2,3)$, which was originally not accepted by either of them.
\end{enumerate}
\end{example}

With the mergeable definition established, we can now formalise the merge function as follows.

\begin{definition}[Merge]
\label{definition:canonical-ps:merge-ps}
Let $sv = \set{(\set{q_a}, b)}, sv' = \set{(\set{q_c}, d)} \in \SV$ be canonical symbolic vectors.
The merge function, noted $merge_{sv}: \SV \times \SV \rightarrow \SVS$, is defined as:
\begin{align*}
    &\text{Given:} \\
    &b' = \set{q_b \in b ~|~ q_c \not\subset_f q_b} \\
    &b'' = \set{q \in \SV ~|~ \forall q_b \in b, \forall q_d \in d, q_c \subseteq_f q_b \land q = join(\set{q_b, q_d})} \\
    &merge_{sv}(sv,sv') = \begin{cases}
    \bm{if}&sv \in \varnothing_{sv} \land sv' \in \varnothing_{sv}:~~~~\varnothing \\
    \bm{if}&sv \in \varnothing_{sv} \land sv' \not\in \varnothing_{sv}:~~~~\set{sv'} \\
    \bm{if}&sv \not\in \varnothing_{sv} \land sv' \in \varnothing_{sv}:~~~~\set{sv} \\
    \bm{if}&mergeable(sv,sv') = \true \land sv \not\in \varnothing_{sv} \land sv' \not\in \varnothing_{sv} \\
    &\begin{cases}
        \bm{if}&q_a \subseteq_f q_c \\   
        &\set{can_{sv}(q_a, b' \cup b'')}\\
        \bm{else}&~merge_{sv}(sv',sv)
    \end{cases} \\
    \bm{if}&mergeable(sv,sv') = \false \land sv \not\in \varnothing_{sv} \land sv' \not\in \varnothing_{sv} \\
    &\set{sv, sv'}
    \end{cases} \\
\end{align*}
\end{definition}

The merging process involves two steps, achieved through the construction of $b'$ and $b''$.
The first step aims to separate vectors that must remain unchanged in $b'$.
This step also addresses the gap scenario where a path is constrained within both symbolic vectors.
Conversely, when vectors from $b$ are included in $q_c$, we combine all of them using $join$ to determine the maximum boundary that satisfies both symbolic vectors.
The final outcome involves taking the minimum tuple of $q_a$ and $q_c$ as the minimum boundary for the new symbolic vector.
Subsequently, $b'$ and $b''$ are assembled as the new excluding set. 
Furthermore, the application of $can_{sv}$ eliminates potential redundancies resulting from the combinations created in $b''$ and ensures a canonical structure.

\begin{lemma}[Merge correctness]
\label{lemma:canonical-ps:merge-ps-canonical}
Let $sv, sv' \in \SV$ be two symbolic vectors.
Then, $\set{sv, sv'} =_{svs} merge_{sv}(sv, sv')$.
\end{lemma}

\begin{proof}
Let us establish the proof for~\cref{lemma:canonical-ps:merge-ps-canonical} by demonstrating that:
\begin{align*}
    &sv = (\set{q_a}, b),~sv' = (\set{q_c}, d)\\
    (1)&~\set{sv,sv'} \setminus_{svs} merge_{sv}(sv,sv') \in \varnothing_{svs}~\land \\
    (2)&~merge_{sv}(sv,sv') \setminus_{svs} \set{sv,sv'} \in \varnothing_{svs} 
\end{align*}

The situations where $sv \in \varnothing_{sv}$, $sv' \in \varnothing_{sv}$, or $mergeable(sv, sv')$ are trivial. 
The only scenario that requires examination is when:
\begin{align*}
    mergeable(sv,sv') = \true \land sv \not\in \varnothing_{sv} \land sv' \not\in \varnothing_{sv}
\end{align*}

Let us concentrate on the case where $q_a \subseteq_f q_c$ (the other case is symmetrical).
We begin by proving case (1) in this configuration.
Both cases can be rewritten as:
\begin{align*}
     &merge_{sv}(sv,sv') = \set{can_{sv}(\set{q_a}, b' \cup b'')} \\
    (1)&~\set{sv,sv'} \setminus_{svs} merge_{sv}(sv,sv') = \\
    (1.1)&~\set{(\set{q_a}, b)} \setminus_{svs}  \set{can_{sv}(\set{q_a}, b' \cup b'')}~\cup_{svs}\\
    (1.2)&~\set{(\set{q_c}, d)} \setminus_{svs}  \set{can_{sv}(\set{q_a}, b' \cup b'')}
\end{align*}

We need to prove that both parts, i.e., (1.1) and (1.2), belong to $\varnothing_{svs}$. (1.1) is trivial; we have the same minimal including vector, and $b \cup b''$ encapsulates $b$.
In fact, $b'$ is a vector from $b$, while $b''$ is the join between $q_b$ and $q_d$, which means that $q_b$ is automatically included in it.
Then, $(1.1) \in_{svs} \varnothing_{svs}$.
(1.2) requires breaking down the term using the difference definition:
\begin{align*}
    (1.2)&~\set{(\set{q_c}, d)} \setminus_{svs}  \set{can_{sv}(\set{q_a}, b' \cup b'')} =\\
    (1.21)&~\set{(\set{q_c}, b \cup \set{q_a})}~\cup_{svs}\\
    (1.22)&~\set{sv_1~|~ q_b \in b', sv_1 = \set{(\set{q_a} \cup \set{q_c} \cup \set{q_b}, d)}}~\cup_{svs}\\
    (1.23)&~\set{sv_2~|~ q_b \in b'', sv_2 = \set{(\set{q_a} \cup \set{q_c} \cup \set{q_b}, d)}}
\end{align*}
In the difference definition (\cref{definition:predicate-structure:difference}), (1.22) and (1.23) appear in only one set.
However, to assist in constructing the proof, we divide the two scenarios.
\begin{itemize}
    \item (1.21): It belongs to $\varnothing_{svs}$, because $q_a \subseteq_f q_c$.
    \item (1.22) $q_b \in b' \Rightarrow q_c \not\subset_f q_b$: By hypothesis, $sv$ and $sv'$ are mergeable.
    Thus, we have:
    \begin{align*}
        (3)&~\forall q_b \in b, q_c \subseteq_f q_b \lor \\
        (4)&~\exists q_d \in d, q_d \subseteq_f join(\set{q_b, q_c}) \\
    \end{align*}
    (3) is not possible, so (4) is true, implying that all $sv_1 \in \varnothing_{sv}$, and thus $(1.22) \in \varnothing_{svs}$.
    \item $(1.23)$: By construction of $b''$, each $q_b \in q_b''$ is the result of the join of $q_b$ and a $q_d$.
    Thus, the corresponding $q_d$ exists in $d$, meaning that $q_d \subseteq_f join(\set{q_b, q_d})$.
    Furthermore, $(1.23) \in \varnothing_{svs}$
\end{itemize}
Therefore, we have proved (1).

Now, let us break down case (2) to prove it:
\begin{align*}
    (2)&~merge_{sv}(sv,sv') \setminus_{svs} \set{sv,sv'} \in \varnothing_{svs} = \\
    &(can_{sv}(\set{q_a}, b' \cup b'') \setminus_{sv} (\set{q_a}, b)) \setminus_{svs} \set{(\set{q_c}, d))}\\
    \\
    (2.1)&~can_{sv}(\set{q_a}, b' \cup b'') \setminus_{sv} (\set{q_a}, b) = \\
    (2.11)&~\set{(\set{q_a}, b \cup \set{q_a})}~\cup_{svs}\\
    (2.12)&~\set{sv_1~|~ q_b \in b, sv_1 = \set{(\set{q_a} \cup \set{q_c} \cup \set{q_b}, b')}}~\cup_{svs}\\
    (2.13)&~\set{sv_2~|~ q_b \in b, sv_2 = \set{(\set{q_a} \cup \set{q_c} \cup \set{q_b}, b'')}}
\end{align*}

\begin{itemize}
    \item $(2.1) \in \varnothing_{svs}$: Trivial
    \item (2.2): $b'$ is constructed using $q_b$ that are not modified.
    Thus, $(2.2) \in \varnothing_{svs}$.
    \item (2.3): $b''$ is built on $join$ on $q_b$ and $q_d$.
    If $q_d \subseteq_f q_b$, then it belongs to $\varnothing_{sv}$.
    If $q_b \subseteq_f q_d$, subtracting by $\set{(\set{q_c}, d))}$ removes the rest by repeating the same process.
\end{itemize}
We have proved case (2).

Cases (1) and (2) result in empty symbolic vector sets, indicating that the union is also empty.
As a result, the property holds.
\end{proof}

\begin{lemma}[Merge reduction]
\label{lemma:canonical-ps:merge-singleton}
Let $sv, sv' \in \SV$ be symbolic vectors.
Then, 
\begin{align*}
    mergeable_{sv}(sv,sv') = \true\Rightarrow merge_{sv}(sv,sv') = \set{sv''},\text{where}~sv'' \in \SV
\end{align*}
\end{lemma}

This lemma indicates that the merging of two symbolic vectors is a singleton when both are mergeable.

\begin{definition}
\label{definition:canonical-ps:shareable}
Let $sv = (a,b), sv' = (c,d) \in \SV$ be two symbolic vectors, and $\set{q_a} = join(a), \set{q_c} = join(c)$ be the simplifications of $a$ and $c$.
Additionally, let $\set{q_{max}} = join(\set{q_a, q_c})$ denote the join vector of $q_a$ and $q_c$.
The shareable function on two symbolic vectors, noted $shareable_{sv}: \SV \times \SV \rightarrow \Bool$, is defined as:
\begin{align*}
    \bm{if}~&sv \in \varnothing_{sv} \lor sv' \in \varnothing_{sv},~shareable_{sv}(sv, sv') = \true \\
    \bm{if}~&sv \not\in \varnothing_{sv} \land sv' \not\in \varnothing_{sv}~~ \\
    shareable_{sv}(sv, sv') &=
    \begin{cases}
    \bm{if}& q_a \leq_f q_c \\
    &\begin{cases}
        \bm{if}& \not\exists q_b \in b, (q_b \subseteq_f q_{max} \land q_b \neq_f q_{max}) \\
        &~\land~\not\exists q_d \in d, q_d \subseteq_{f} q_{max} \\
        &\true\\
        \bm{else}& \false
    \end{cases}\\
    \bm{else}& shareable_{sv}(sv',sv)
    \end{cases}
\end{align*}
\end{definition}

The function $shareable_{sv}$ is designed to determine when two symbolic vectors have a shareable component.
Checking shareability involves examining whether there exists a gap between two symbolic vectors.
By building $q_{max}$, we seek the first value that is potentially shared between two symbolic vectors and will then be the minimum vector of the shared symbolic vector. 
Therefore, if a tuple of $b$ or $d$ is lower than $q_{max}$, this implies that sharing is not possible.
This is similar to having a gap between two intervals such as $[1,4]$ and $[7,9]$.

When the function returns true, this means that we can extract a subset from one of the symbolic vectors and transfer it to the other.
Although the shared component could be moved to either side, to ensure future canonicity, we only allow the merge to be performed on one side due to the lexicographic ordering, making the process deterministic and terminating.

As illustrated in~\Cref{figure:canonical-ps:shearbility11}, a scenario arises in which two symbolic vectors are shareable.
In this case, the shareability should be considered $\true$ only when the shared component is not on the correct symbolic vector. 
The condition $q_b \neq_f q_{max}$ guarantees this, ensuring that the vector with the lowest lexicographic order does not contain the constraint.
In other words, only the largest symbolic vector will be constrained by $q_{max}$.
This prevents overlap between the two structures.

The purpose of $shareable_{sv}$ extends beyond the mentioned scenario. These conditions also encompass situations of \textit{overlapping} and \textit{mergeability}.
In essence, when the function evaluates to true, it signifies that the symbolic vector set is non-canonical.

\begin{example}
Let us observe some examples of its application.
    \begin{enumerate}
        \item $shareable_{sv}((\set{(1,1)}, \set{(3,3)}), (\set{(5,5)}, \set{(7,7)})) = \false. \\$
        Indeed, there exists $q_b = (3,3)$ such that $(3,3) \subseteq_f (5,5) \land (3,3) \neq_f (5,5)$.
        In this context, this case is similar to a simple comparison of intervals such as $[1,3)$ and $[5,7)$, where it seems clear that the interval $[1,3)$ and $[5,7)$ do not overlap.
        \item $shareable_{sv}((\set{(1,1)}, \set{(5,5)}), (\set{(5,5)}, \set{(7,7)})) = \true. \\$
        The merge case is handled.
        \item $shareable_{sv}((\set{(1,1)}, \set{(9,9)}), (\set{(5,5)}, \set{(7,7)})) = \true. \\$
        The overlap case is also handled.
        \item $shareable_{sv}((\set{(2,4)}, \set{(4,4)}), (\set{(4,2)}, \varnothing)) = \true. \\$
        In this case, there is neither an overlap nor a potential merge.
        However, the constraint $(4,4)$ is not on the correct symbolic vector.
        \item $shareable_{sv}((\set{(2,4)}, \varnothing), (\set{(4,2)}, \set{(4,4)})) = \false. \\$
        This time, $(4,4)$ is on the appropriate symbolic vector with respect to $\leq_f$.
\end{enumerate}
\end{example}

\begin{lemma}
\label{lemma:mergeable-imply-shareable}
Let $sv, sv' \in \SV$ be two symbolic vectors.
Then,
\begin{align*}
    mergeable(sv,sv') = \true \Rightarrow shareable_{sv}(sv,sv') = \true
\end{align*}
\end{lemma}

\begin{proof}
    To prove this property, we observe each case where $mergeable(sv,sv')$ is $true$, and then show that it is also $true$ for $shareable_{sv}(sv,sv')$.
    Let us assume that $sv = (a,b), sv' = (c,d),~q_{max} = join(\set{q_a, q_c}),~q_{max}' = join(\set{q_b, q_c})$.
    
    When $sv \in \varnothing_{sv}$ (or $sv' \in \varnothing_{sv}$) is true, the implicitation is automatically $true$.
    
    Let us observe the case where $q_a \subseteq_f q_c$.
    First, $q_a \subseteq_f q_c \Rightarrow q_a \leq_f q_c$.
    Therefore, we have to show that if the condition $\forall q_b \in b, q_c \subseteq_f q_b \lor \exists q_d \in d, q_d \subseteq_f join(\set{q_b, q_c})$ holds, then the condition $\not\exists q_b \in b, (q_b \subseteq_f q_{max} \land q_b \neq_f q_{max}) \land \not\exists q_d \in d, q_d \subseteq_{f} q_{max}$ also holds.
    We divide this sub-proof into two parts:
    \begin{itemize}
        \item Case where $\forall q_b \in b, q_c \subseteq_f q_b$:
        because $q_a \subseteq_f q_c$, then $q_{max} \subseteq_f q_b$, which means that we cannot break the condition $q_b \subseteq q_{max} \land q_b \neq_f q_{max}$.
        
        In addition, $q_a \subseteq_f q_c \Rightarrow q_{max} = q_c$.
        Then, we can deduce that $q_d \subseteq_f q_{max} \Rightarrow q_d \subseteq_f q_c \Rightarrow sv' \in \varnothing_{sv}$.
        Therefore, the only way that $q_d \subseteq_f q_{max}$ to exist is for $sv'$ to be empty, which is not possible under condition $sv' \not\in \varnothing_{sv}$.
        As a result, $\not\exists q_d \in d, q_d \subseteq_f q_{max}$ also holds.
        \item Case where $\exists q_d \in d, q_d \subseteq_f join(\set{q_b, q_c})$: let us show this case for the condition $\not\exists q_b \in b, (q_b \subseteq_f q_{max} \land q_b \neq_f q_{max})$.
        Due to $q_a \subseteq_f q_c$, we know that $q_{max} = q_c$. 
        We are supposed to have $q_d \subseteq_f q_{max}' = q_d \subseteq_f join(\set{q_b, q_c})$.
        However, we know that $\forall q_d \in d, q_c \subseteq_f q_d$.
        Furthermore, if $q_b \subseteq_f q_{max} = q_b \subseteq_f q_c$, the only way to satisfy the inequality $q_d \subseteq_f join(\set{q_b, q_c})$ is to have $q_d = join(\set{q_b, q_c})$.
        In addition, if we suppose that $q_b =_f q_{max}$, it implies that $q_b = q_c$, and then $q_c = q_d$. 
        Having $q_c = q_d$ means that $sv' \in \varnothing_{sv}$, which is not possible due to the initial conditions.
        Therefore, there does not exist a solution that breaks the first part of the condition.
        
        Let us prove the second part of the condition, i.e. $q_{max} \subseteq_f q_d \land q_{max} \neq q_d$.
        We have two cases, with $q_{max} \subseteq_f q_d \land q_{max} \neq q_d$ and $q_d \subseteq_f q_{max}$.
        If $q_{max} \subseteq_f q_d \land q_{max} \neq q_d$, condition $\not\exists q_d \in d, q_d \subseteq_{f} q_{max}$ is satisfied.
        For $q_d \subseteq_f q_{max}$, due to $q_{max} = q_c$, we obtain $q_d \subseteq_f q_c$, which means that $sv'$ is an empty symbolic vector.
        When $sv'$ is empty, we know that this is not possible by the condition $sv' \not\in \varnothing_{sv}$.
        Therefore, the condition remains $true$.
    \end{itemize}
    Consequently, the property is valid.
\end{proof}

\subsection{Canonicity of symbolic vector sets}

From the previous sections, we now have all the necessary ingredients to formally define the canonical form of a symbolic vector set, as well as the canonical operations on it.

\begin{definition}[Canonicity of symbolic vector sets]
\label{definition:canonical-ps:canonical-svs}
Let $svs \in \SVS$ be a symbolic vector set.
$svs$ is canonical if and only if:
\begin{enumerate}
    \item $\forall sv \in svs, sv$ is canonical.
    \item $svs \not\in \varnothing_{svs} \lor svs = \varnothing$
    \item $\forall sv,sv' \in svs, sv \neq sv', shareable_{sv}(sv,sv') = \false$
\end{enumerate}
\end{definition}

Based on this definition, three conditions must be met to guarantee canonicity.
The initial condition stipulates that all elements within a symbolic vector set must be canonical.
The second condition requires a unique form for the empty symbolic vector set.
Lastly, the condition we examined previously becomes significant.
When $shareable_{sv}$ is evaluated to be true, a portion of the largest symbolic vector can be transferred to the smaller one between $sv$ and $sv'$.
However, it is important to note that the shared part might already partially or entirely exist within the smaller symbolic vector.
This could occur in cases where one symbolic vector overlaps another.
In such instances, the larger symbolic vector would be either entirely removed or adjusted to exclude the portion already present in the smaller symbolic vector.
Note that, as shown in \cref{lemma:mergeable-imply-shareable}, whenever the mergeability condition holds, the shareability condition also holds. Therefore, this case is already covered by the definition of shareability and need not be explicitly included as a canonicity condition.

\begin{theorem}[Unicity of the representation]
\label{theorem:canonical-ps:canonicity-svs}
Let $svs, svs' \in \SVS$ be two canonical symbolic vector sets.
Then, $\usf(svs) = \usf(svs') \Leftrightarrow svs = svs'$.
\end{theorem}

\begin{remark}
For the sake of brevity and simplicity, we omit the formal description of canonical homomorphisms for $\cup_{svs}$, $\cap_{svs}$, $\setminus_{svs}$, and $\neg_{svs}$.
We assume their existence because, formally, they can be viewed as a canonisation of classical set operations and are intended to be used implicitly.
Furthermore, the complete theory is described in~\cite{morard:2024:symbolic}, including additional definitions and properties, as well as all the canonical operations and proofs related to them.
\end{remark}


The upcoming proof is the most intricate among all the theories introduced, necessitating the introduction of multiple definitions and properties for clarity.
Note that the final proof can be found at the end of this subsection.

\begin{definition}[Gap] 
    Let $sv =(\{q_a\}, b) \in \SV, sv'=(\{q_c\}, d) \in \SV$ be canonical, and $\{q_{max}\}=join(\{q_a,q_c\})$.
    The gap relation between two symbolic vectors, denoted as $gap: \SV \times \SV$, is defined as follows:
    \begin{align*}
        gap(sv, sv') \Leftrightarrow&~\exists q_b \in b, (q_b \subseteq_f q_{max} \land q_b \not =_f q_{max}) \\
        \lor&~\exists q_d \in d, (q_d \subseteq_f q_{max} \wedge q_d \not =_f q_{max}) \\
        \lor&~(q_{max} \in b \land q_{max} \in d) 
    \end{align*} 
\end{definition}

As a corollary, we have the negation:
\begin{corollary}[$\neg$Gap a.k.a Joinable]
 \begin{align*}
\neg gap(sv, sv') \Leftrightarrow&~\nexists q_b \in b, (q_b \subseteq_f q_{max} \land q_b \neq_f q_{max}) \\
\land&~\nexists q_d \in d, (q_d \subseteq_f q_{max} \land q_d \neq_f q_{max}) \\
\land&~(q_{max} \not\in b \lor q_{max} \not\in d)
\end{align*}   
\end{corollary}

The $gap$ is a way of expressing continuity between the underlying sets of symbolic vectors.
Imagine a grid of dimension $n$ where our vectors exist as nodes. 
Movement in this grid is only allowed between adjacent nodes.
The gap case represents two sets of nodes/vectors such that it is not possible to move from one set to the other while adhering to these rules and staying within one set at each step.
There is a gap, a space separating the underlying sets of the symbolic vectors.

This is why the situation of $\neg gap$ is more intuitively named $joinable$ since it refers to the situation where the two underlying sets can be joined.
Please be cautious, as "\textit{joinable}" differs from determining whether two symbolic vectors are mergeable. 
Furthermore, being mergeable implies being joinable, but the reverse is not necessarily true.

\begin{definition}[Overlap]
    Let $sv=(\{q_a\}, b) \in \SV, sv'=(\{q_c\}, d) \in \SV$ be canonical symbolic vectors, and $\{q_{max}\} = join(\{q_a, q_c\})$.
    The overlap relation between two symbolic vectors, denoted as $overlap: \SV \times \SV$, is defined as:
    \begin{align*}
        overlap(sv, sv') \Leftrightarrow&~(\forall q_b \in b, q_b \not\subset_f q_{max} \lor (q_{max} \subseteq_f q_b \land q_{max} \neq_f q_b))\\ 
        \land&~(\forall q_d \in d, q_d \not\subset_f q_{max} \lor (q_{max} \subseteq_f q_d \land q_{max} \neq_f q_d))
    \end{align*}
\end{definition}

It is noted that if $q_{max} \subseteq_f q \land q_{max} \neq_f q$ is true, then it always verifies that $q \not\subset_f q_{max}$. 
However, it is important to notice that the reciprocal implication is not true.
Indeed, it is recalled that $q \not\subset_f q_{max}$ implies $q \neq_f q_{max}$ but does not imply $q_{max} \subseteq_f q$ as it is possible that $q$ and $q_{max}$ are different and simply non-comparable.

The following example helps visualize the kind of non-trivial cases captured by the definition. 
Let $sv=(\{(1,0)\}, \{(2,0)\}), sv'=(\{(0,2)\}, \{(0,3)\})$ where we have $q_{max} = \set{(1,2)}$.
In this example, the two vectors overlap only at $(1,2)$.
However, we can indeed verify $overlap(sv, sv')$.
This example highlights the need to consider the non-comparability of exclusion bounds with $q_{max}$.

\begin{corollary}[$\neg$Overlap] 
\label{neg_overlap}
    \begin{align*}
    \neg overlap(sv, sv') \Leftrightarrow (\exists q_b \in b, q_b \subseteq_f q_{max}) \lor (\exists q_d \in d, q_d \subseteq_f q_{max}) 
    \end{align*}
\end{corollary}
\begin{proof}
    \begin{align}
        \neg overlap(sv, sv') &\Leftrightarrow (\exists q_b \in b, \neg(q_b \not\subset_f q_{max} \lor (q_{max} \subseteq_f q_b \land q_{max} \neq_f q_b))) \\
        &\lor (\exists q_d \in d, \neg(q_d \not\subset_f q_{max} \lor (q_{max} \subseteq_f q_d \land q_{max} \neq_f q_d))) \nonumber\\\nonumber\\
       &\Leftrightarrow (\exists q_b \in b, q_b \subseteq_f q_{max} \land (q_{max} \not\subset_f q_b \lor q_{max} =_f q_b)) \\
       &\lor (\exists q_d \in d, q_d \subseteq_f q_{max} \land (q_{max} \not\subset_f q_d \lor q_{max} =_f q_d)) \nonumber\\\nonumber\\
       &\Leftrightarrow (\exists q_b \in b, q_b \subseteq_f q_{max}) \lor (\exists q_d \in d, q_d \subseteq_f q_{max}) 
   \end{align}   
   To understand the steps, it is recalled that $\neg(x \not\subset_f y) \equiv (x\subseteq_f y)$.
   For the transition from step 5 to step 6, once distributed, we obtain an expression equivalent to $q \subseteq_f q_{max}$, hence the simplification.
\end{proof}

As a corollary, we have that $overlap(sv, sv') \implies \neg gap(sv, sv')$.
We observe that $q_{max}$ is the potential shared point.
If we are ensured that each exclusion bound is non-comparable ($\not\subset_f$) or strictly greater than $q_{max}$, then we know that shared vectors exist.
We have thus formalised the notion of empty or non-empty intersection in terms of symbolic vectors.

Furthermore, we observe:
\begin{lemma} 
\label{overlap_equiv_internonempty}
    Let $sv, sv' \in \SV$ be canonical, then 
    \begin{align*}
        overlap(sv, sv') \iff uf(sv) \cap uf(sv') \neq \varnothing
    \end{align*}
\end{lemma}

\begin{proof}
\hfill

\begin{itemize}
\item[\boxed{\implies}] Let $sv=(\{q_a\}, b), sv'=(\{q_c\}, d)$ be canonical such that $overlap(sv, sv')$. 
Let us set $\{q_{max}\}=join(\{q_a, q_c\})$. 
By definition, $q_a \subseteq_f q_{max}$, and due to the fact that $overlap(sv, sv')$, we know that all exclusion bounds are either greater than $q_{max}$ or non-comparable.
Thus, $q_{max} \in uf(sv)$.
Applying the same reasoning to $sv'$, we can conclude that
\begin{align*}
    &~q_{max}\in uf(sv) \land q_{max}\in uf(sv') \\
    \iff&~q_{max} \in uf(sv) \cap uf(sv') \implies uf(sv) \cap uf(sv') \neq \varnothing
\end{align*}
\item[\boxed{\impliedby}] By contrapositive, that is, the proof of $\neg overlap(sv, sv') \implies uf(sv) \cap uf(sv') = \varnothing$.
Let $sv=(\{q_a\}, b), sv'=(\{q_c\}, d)$ be canonical such that $\neg overlap(sv, sv')$.
Let $\{q_{max}\}=join(\{q_a, q_c\})$ and $\Tilde{sv}= sv~\cap~sv'=(\{q_{max}\}, b\cup d)$.
From \Cref{neg_overlap}, we know that there is an exclusion bound $q_{excl} \in b\cup d$ such that $q_{excl} \subseteq_f q_{max}$. 
As a direct consequence, we have $uf(\Tilde{sv})=\varnothing$.
Since $uf(\Tilde{sv}) = uf(sv\cap sv') = uf(sv)\cap uf(sv')$, we conclude that $uf(sv)\cap uf(sv') = \varnothing$.
\end{itemize}
\end{proof}

\begin{definition}[Adjacent] 
\label{adjacence}
    Let $sv=(\{q_a\}, b) \in \SV, sv'=(\{q_c\}, d) \in \SV$ be canonical and $\{q_{max}\}=join(\{q_a,q_c\})$.
    \begin{align*}
        adjacent(sv, sv') \iff&~(q_{max} \in b \lor q_{max} \in d) \land~(q_{max}\not\in b \lor q_{max}\not\in d)
    \end{align*}
\end{definition}
This definition states that adjacency corresponds to the case where two vectors are joinable and do not overlap.

\begin{lemma}
\label{non_share_implies_non_overlap}
    Let $sv, sv' \in \SV$ be canonical, then
    \begin{align*}
        \neg shareable(sv, sv') \implies \neg overlap(sv, sv')
    \end{align*}
\end{lemma}

\begin{proof}
Let $sv,sv' \in \SV$, where $sv=(\{q_a\}, b), sv'=(\{q_c\}, d)$.
Let us assume that $\neg shareable(sv, sv')$ and $\{q_{max}\}=join(\{q_a, q_c\})$.
Due to the existence of lexicographic order, $shareable$ is always well-defined, as well as its negation.
\begin{align*}
    \neg shareable(sv, sv') \iff \overbrace{\exists q_b \in b, (q_b \subseteq_f q_{max} \land q_b \neq_f q_{max}) \lor \exists q_d \in d, q_d \subseteq_f q_{max}}^{\text{case : } q_a \leq_f q_c} \\
    \lor \underbrace{\exists q_d \in d, (q_d \subseteq_f q_{max} \land q_d \neq_f q_{max}) \lor \exists q_b \in d, q_b \subseteq_f q_{max}}_{\text{case : } \neg(q_a \leq_f q_c)}
\end{align*}
We recall that $\neg overlap(sv, sv') \Leftrightarrow (\exists q_b \in b, q_b \subseteq_f q_{max}) \lor (\exists q_d \in d, q_d \subseteq_f q_{max})$.
Trivially, each case in the disjunction of $\neg shareable(sv, sv')$ implies one of the cases in $\neg overlap(sv, sv')$. 
Therefore, we have successfully demonstrated $\neg shareable(sv, sv')\implies \neg overlap(sv, sv')$.
\end{proof}
As a direct corollary of this result, we have $\neg shareable(sv, sv') \implies uf(sv)\cap uf(sv') = \varnothing$.

We will now introduce the two main lemmas that are at the core of the proof of uniqueness.

\begin{lemma}[Total decomposition implies shareability]
\label{lemma:total-decomposition-shareability}
Let $sv=(\{q_a\},b)\in\SV$ be a non-empty canonical symbolic
vector, and let
\[
    \widetilde{svs}
    =
    \{sv_1,\dots,sv_k\},
    \qquad k\geq 2,
\]
be a finite set of non-empty canonical symbolic vectors. If
\[
    uf(sv)
    =
    \bigcup_{sv_i\in\widetilde{svs}}uf(sv_i),
\]
then there exist two distinct symbolic vectors
$sv_i,sv_j\in\widetilde{svs}$ such that
\[
    shareable_{sv}(sv_i,sv_j)=\true.
\]
\end{lemma}

This lemma can be interpreted as stating that if the set of vectors represented by several symbolic vectors $sv_i$ can also be represented by a single symbolic vector, then the $sv_i$ can be replaced by that symbolic vector.

\begin{proof}
Let
\[
    sv_i=(\{q_i\},b_i)
\]
for every $sv_i\in\widetilde{svs}$.

We first observe that $uf(sv)$ is downward closed above its minimum
bound $q_a$. More precisely, if $x\in uf(sv)$ and
\[
    q_a\subseteq_f y\subseteq_f x,
\]
then $y\in uf(sv)$. Indeed, if some $q_b\in b$ satisfied
$q_b\subseteq_f y$, transitivity would give $q_b\subseteq_f x$,
contradicting $x\in uf(sv)$.

Suppose first that two distinct symbolic vectors
$sv_i,sv_j\in\widetilde{svs}$ overlap. By
Lemma~\ref{non_share_implies_non_overlap}, used by contraposition,
\[
    overlap(sv_i,sv_j)
    \Longrightarrow
    shareable_{sv}(sv_i,sv_j)=\true,
\]
and the result follows.

Assume now that the underlying sets of the elements of
$\widetilde{svs}$ are pairwise disjoint.

Since $sv$ is non-empty and canonical, its minimum bound satisfies
$q_a\in uf(sv)$. Therefore, there exists a unique
$sv_0=(\{q_0\},b_0)\in\widetilde{svs}$ such that
\[
    q_a\in uf(sv_0).
\]
Since $q_a\in uf(sv_0)$, we have $q_0\subseteq_f q_a$. Conversely,
$q_0\in uf(sv_0)\subseteq uf(sv)$, and hence
$q_a\subseteq_f q_0$. Consequently,
\[
    q_0=q_a.
\]

Among the symbolic vectors in
$\widetilde{svs}\setminus\{sv_0\}$, choose
$sv_i=(\{q_i\},b_i)$ whose minimum bound $q_i$ is minimal with
respect to $\subseteq_f$. Such an element exists because
$\widetilde{svs}$ is finite.

We show that
\[
    \nexists q_{b_0}\in b_0,\quad
    q_{b_0}\subseteq_f q_i
    \land q_{b_0}\neq_f q_i.
\]
Assume otherwise that such a $q_{b_0}$ exists. Since $sv_0$ is
canonical,
\[
    q_a=q_0\subseteq_f q_{b_0}\subseteq_f q_i.
\]
Moreover, $q_i\in uf(sv_i)\subseteq uf(sv)$. By the downward-closure
property established above, this implies
\[
    q_{b_0}\in uf(sv).
\]
Hence, some $sv_j=(\{q_j\},b_j)\in\widetilde{svs}$ contains
$q_{b_0}$. We have $sv_j\neq sv_0$, since $q_{b_0}$ is an exclusion
bound of $sv_0$, and
\[
    q_j\subseteq_f q_{b_0}\subset_f q_i.
\]
This contradicts the minimality of $q_i$.

Furthermore,
\[
    \nexists q_{b_i}\in b_i,\quad q_{b_i}\subseteq_f q_i.
\]
Indeed, the canonicity of $sv_i$ gives
$q_i\subseteq_f q_{b_i}$ for every $q_{b_i}\in b_i$. Thus,
$q_{b_i}\subseteq_f q_i$ would imply $q_{b_i}=q_i$, making $sv_i$
empty, contrary to the hypothesis.

Finally, $q_0=q_a\subseteq_f q_i$, and therefore
$q_0\leq_f q_i$ in the lexicographic order. Since
\[
    join(\{q_0,q_i\})=\{q_i\},
\]
the two properties established above are exactly the conditions in
the definition of $shareable_{sv}$. Consequently,
\[
    shareable_{sv}(sv_0,sv_i)=\true.
\]

Thus, in all cases, $\widetilde{svs}$ contains two distinct
shareable symbolic vectors.
\end{proof}

\begin{lemma}[Non-partial decomposition] 
\label{lemma_non_depassement}
    Let $svs, svs' \in \SVS$ be two canonical symbolic vector sets such that $\usf(svs)=\usf(svs')$, $\usf(svs) \neq \varnothing$ and $\usf(svs') \neq\varnothing$.
    The non-partial decomposition property is defined as:
    \begin{align*}
        &\forall \Tilde{svs} \subset svs', |\Tilde{svs}|\geq 2,\\ &\nexists sv \in svs, \left(sv \not\in \Tilde{svs} \land uf(sv) \subseteq \bigcup_{sv'\in\Tilde{svs}} uf(sv') \land \forall sv'\in\Tilde{svs}, overlap(sv, sv')\right) 
    \end{align*}
\end{lemma}

\begin{proof}
Let $svs, svs' \in \SVS$  be canonical such that $\usf(svs)=\usf(svs')$, $\usf(svs) \neq \varnothing$ and $\usf(svs') \neq \varnothing$. 
By reductio ad absurdum, let us assume the existence of such a set noted as $\Tilde{svs}\subset svs'$.
In addition, assume that $sv=(\{q_c\}, d)\in svs$ such that $sv\not\in\Tilde{svs}$, $uf(sv)\subseteq\bigcup_{sv'\in\Tilde{svs}}uf(sv')$, and $\forall sv'\in\Tilde{svs}, overlap(sv ,sv')$.

First, let us consider the case of a total decomposition, i.e.,
\[
    uf(sv)
    =
    \bigcup_{sv'\in\Tilde{svs}} uf(sv').
\]
By Lemma~\ref{lemma:total-decomposition-shareability}, there exist
two distinct symbolic vectors $sv_i,sv_j\in\Tilde{svs}$ such that
\[
    shareable_{sv}(sv_i,sv_j)=\true.
\]
Since $\Tilde{svs}\subset svs'$, both $sv_i$ and $sv_j$ belong to
$svs'$. This contradicts the canonicity of $svs'$ $\lightning$.

Next, let us address the case of non-partial decomposition where it is not total.
This case is intricate and requires the use of lexicographic ordering. 
Since lexicographic ordering is a total order on vectors, we can order the symbolic vectors based on their inclusion bounds, as these are unique for each canonical symbolic vector.

Let us arbitrarily assume that we are working with the "smallest" one, in the lexicographic sense, among the partial decomposition.
We have assumed that there is at least one, so we can always make this assumption.
Let $sv=(\{q_c\}, d)\in svs$ be this symbolic vector such that $uf(sv)\subset \bigcup_{sv'\in\Tilde{svs}}uf(sv')$ and remains valid for the rest of the hypothesis.
In addition, $\forall \hat{sv}=(\{\hat{q_a}\}, \hat{b})\in svs$ such that  
\begin{align*}
\hat{sv} \not\in \Tilde{svs} \land uf(\hat{sv}) \subseteq \bigcup_{sv'\in\Tilde{svs}} uf(sv') \land \forall sv'\in\Tilde{svs}, overlap(\hat{sv}, sv')
\end{align*}
then
\begin{align} \label{minimalite}
q_c \leq_f \hat{q_a}   
\end{align}

Since $uf(sv)\subset \bigcup_{sv'\in\Tilde{svs}}uf(sv')$, we know that there is at least one $sv'\in\Tilde{svs}$ that is not covered entirely by $uf(sv)$.
More precisely, there exists an $sv'$ such that $uf(sv)\cap uf(sv')\neq uf(sv')$.
Among the vectors in $sv'$ that are not covered by $sv$, there must be at least one: $a\in (uf(sv')/uf(sv))$ that is adjacent to $sv$.
If this were not the case, we would have $sv'$ that contains a gap, which is impossible.
Hence, let us consider $sv_a=(\{q_a\}, b) \in svs$ as a symbolic vector such that $a \in uf(sv_a)$.
Thus, we have $a \in uf(sv_a)\land a\in uf(sv')$ and $\{q_{max}\}=join(\{q_a, q_c\})$.

Now, by the canonicity of $svs$, we know that $\neg shareable(sv, sv_a)$. This corresponds to the following two cases:
\begin{itemize}
    \item If $q_a\leq_f q_c$, this can be broken down into two cases:
    \begin{itemize}
        \item $(q_a \subseteq_f q_c\land q_a\neq_f q_c)$ or $q_a$ is not comparable with $q_c$: it means that $sv_a$ exists outside of $sv'$.
        Since $a\in uf(sv_a)$ and $a\in uf(sv')$, we have
        \[
            a\in uf(sv_a)\cap uf(sv').
        \]
        Therefore,
        \[
            uf(sv_a)\cap uf(sv')\neq\varnothing.
        \]
        By Lemma~\ref{overlap_equiv_internonempty}, it follows that
        \[
            overlap(sv_a,sv').
        \]
        However, this contradicts the minimality of $sv$
        (\Cref{minimalite}) $\lightning$.
        \item $q_a =_f q_c$: From this assumption, $q_a \in sv\land q_a\in sv_a \Rightarrow uf(sv)\cap uf(sv_a)\neq\varnothing\Rightarrow overlap(sv, sv_a)$.
        By contraposition (\Cref{non_share_implies_non_overlap}), it implies that $shareable(sv, sv_a)$ is true, which is not possible by hypothesis $\lightning$.
    \end{itemize}
    \item If $q_c \leq_f q_a$: we immediately rule out the cases of equality and non-total comparability that we just addressed above.
    So $q_c \subseteq_f q_a$, which implies $q_{max}=q_a$.
    Therefore, we are in the case $\exists q_d \in d, (q_d \subseteq_f q_{max}\land q_d\subseteq_f q_{max}) \lor \exists q_b \in b, q_b \subseteq_f q_{max}$.
    \begin{itemize}
        \item If $\exists q_d \in d, (q_d \subseteq_f q_{max}\land q_d\subseteq_f q_{max})$, then we observe that this implies $gap(sv, sv_a)$, which is absurd because $sv_a$ contains $a$, which is a vector adjacent to $sv$, and therefore, $sv_a$ is necessarily adjacent to $sv$ $\lightning$.
        \item If $\exists q_b \in b, q_b \subseteq_f q_{max}$, then $\exists q_b \in b, q_b \subseteq_f q_a$, leading to $uf(sv_a)=\varnothing$ that is impossible because of the canonicity of $svs$ $\lightning$. 
    \end{itemize}
\end{itemize}
Like all cases are absurd, we know that $shareable(sv, sv_a)$.
As a result, this contradicts the assumption of canonical $svs$ $\lightning$.

Consequently, all scenarios are absurd, and we have thus demonstrated that the non-partial decomposition cannot exist. 
\end{proof}

\begin{lemma}[Non-inclusion] 
\label{lemma_non_inclusion}
Let $svs, svs' \in \SVS$ be two canonical symbolic vector sets such that $\usf(svs)=\usf(svs')$, $\usf(svs) \neq \varnothing$ and $\usf(svs')\neq\varnothing$. Then,
\begin{align*}
    \nexists sv \in svs, \exists sv' \in svs', \ \text{such that}\ uf(sv) \subset uf(sv')
\end{align*}
\end{lemma}

\begin{proof}
Let $svs, svs' \in \SVS$  be canonical such that $\usf(svs)=\usf(svs')$, $\usf(svs) \neq \varnothing$ and $\usf(svs') \neq \varnothing$. 

By reductio ad absurdum, assume $\exists sv\in svs, \exists sv'\in svs'$ such that $uf(sv)\subset uf(sv')$.
The situation $\usf(svs)=\usf(svs')$ implies the existence of one or more vectors within $svs$ to cover $uf(sv')$. 
So, $\exists sv_1, \cdots, sv_n \in svs$, and we set $\Tilde{svs}={sv, sv_1, \cdots, sv_n}$ such that
$$uf(sv) \cup \bigcup_i^n uf(sv_i) = \bigcup_{\Tilde{sv}\in\Tilde{svs}}uf(\Tilde{sv}) = uf(sv')$$
We can achieve an exact cover because we have demonstrated~\Cref{lemma_non_depassement}, which ensures that each $uf(sv)$ is either equal to or contained within a $uf(sv')$.
Thanks to it, it remains to prove that $sv'\not\in\Tilde{svs}$ and $\forall \Tilde{sv}\in\Tilde{svs}, overlap(sv', \Tilde{sv})$.

By reductio ad absurdum, let us show $sv'\not\in\Tilde{svs}$, i.e. by assuming that $sv'\in\Tilde{svs}$, which also implies that $sv' \in svs$.
Suppose that there exists $sv_x \in \Tilde{svs}$ with $sv_x \neq sv'$. 
Since $uf(sv') \cup uf(sv_x) = uf(sv')$ and $sv_x$ cannot be empty due to the canonicity of $svs$, we have $uf(sv') \cap uf(sv_x) \neq \varnothing$.
By \Cref{overlap_equiv_internonempty}, we know $overlap(sv', sv_x)$, which, by the contrapositive of \Cref{non_share_implies_non_overlap}, gives us $shareable(sv', sv_x)$. 
This contradicts the assumption of canonicity of $svs$ $\lightning$.

In a similar way, let us show $\forall \Tilde{sv}\in\Tilde{svs}, overlap(sv', \Tilde{sv})$.
We assume that there exists $\Tilde{sv}_x\in \Tilde{svs}$ such that $\neg overlap(sv', \Tilde{sv}_x)$.
Since the union is equal to $uf(sv')$, this means that $\Tilde{sv}_x$ contributes nothing to the union.
Consequently, $\Tilde{sv}_x \in \varnothing_{sv}$, contradicting the assumption of a canonical $svs$ $\lightning$.

Thus, we have all the assumptions to apply~\Cref{lemma_non_depassement}, which ensures that such a decomposition is not possible.
\end{proof}

We are now able to introduce the complete proof of~\Cref{theorem:canonical-ps:canonicity-svs}, which states the following:
Let $svs, svs' \in \SVS$ be two canonical symbolic vectors.
Then, $\usf(svs) = \usf(svs') \Leftrightarrow svs = svs'$.

\begin{proof}
\hfill
\begin{itemize}
\item[\boxed{\impliedby}] By the determinism of $\usf$, if $svs = svs'$, then $\usf(svs) = \usf(svs')$ is obviously true.
\item[\boxed{\implies}]
Let us prove this by induction on the size of the underlying set.
Let us start with the base case, where $|\usf(svs)|=|\usf(svs')|=0$.
We know that there is only one representation of an underlying set of size $0$, and that is $\varnothing_{svs}$. 
If $svs$ or $svs'$ contained any $sv$, by the second rule of the canonicity definition, the set of corresponding symbolic vectors would no longer be empty and its underlying set would not be of size 0.

Now let us consider the induction step. 
Assume canonical $svs, svs'$ such that $\usf(svs)=\usf(svs')$. 
Let $|svs|=n$ and $|svs'|=m$, representing the number of symbolic vectors in each set. Without loss of generality, we can assume that $n\leq m$ since the cases are symmetric.
\begin{align*}
    \usf(svs) =& \bigcup_{sv\in svs}uf(sv) = \bigcup_{sv'\in svs'}uf(sv') = \usf(svs')
\end{align*}
At each step, we address the same questions regarding the underlying sets of the considered symbolic vectors.
Specifically, we inquire whether each underlying set of an $sv\in svs$ is also an underlying set of an $sv'\in svs'$.

Now, let us consider the two cases.
\begin{itemize}
\item[\textbf{Cas $n=m$} : ] 
We distinguish three cases regarding the equality of underlying sets:
\begin{enumerate}
    \item 
    (One $uf(sv)$ is "larger" than its counterpart) There is at least one $sv$ such that $uf(sv) \subset \bigcup_{\Tilde{sv}\in\Tilde{svs}}uf(\Tilde{sv}) = \usf(\Tilde{svs})$, where $\Tilde{svs} \subset svs'$ and $|\Tilde{svs}| \geq 2$.
    We know that this case is impossible by \Cref{lemma_non_depassement}(Non-partial decomposition). $\lightning$ 
    \item (One $uf(sv)$ is "smaller" than its counterpart) There is at least one $sv$ corresponding to the case where $\exists sv' \in svs'$ such that $uf(sv) \subset uf(sv')$. 
    We know that this case is impossible by \Cref{lemma_non_inclusion}(Non-Inclusion). $\lightning$
    \item Each $uf(sv)$ is equal to a $uf(sv')$, i.e., $\forall sv, \exists ! sv', uf(sv)=uf(sv')$, which is equivalent to saying $\forall sv, \exists ! sv', sv=sv'$ due to the uniqueness of the representation of canonical symbolic vectors.
So, $svs=svs'$ because they are composed of the same canonical symbolic vectors. 
\end{enumerate}

\item[\textbf{Cas $n<m$} : ] 
Let us distinguish the same cases as for $n=m$. By the exact same reasoning, we can conclude that cases $1.$ and $2.$ are still impossible.
Now, let us verify that the third case also does not hold.
\begin{enumerate}
    \item[3.] As each $sv \in svs$ is equal to a unique $sv' \in svs'$, we can construct the following non-empty sets:  $svs'_e=\{sv'\in svs' \mid sv'\not\in svs\}$ and $svs'_f=\{sv'\in svs'\mid sv'\in svs\}$.
    We observe that
    \begin{align*}
        \usf(svs)=\usf(svs')=\usf(svs'_e)\cup \usf(svs'_f)
    \end{align*}
    This implies, due to $\usf(svs) = \usf(svs'_f)$ that $\usf(svs'_e) = \varnothing$.
    In addition, this results in symbolic vectors in $svs'_e$ being empty, which originally belongs to $svs'$.
    This contradicts the assumption that $svs'$ is canonical $\lightning$.
\end{enumerate}
\end{itemize}
Through all these cases, all are impossible except one.
Thus, we have shown that if $svs$ and $svs'$ are canonical, and $\usf(svs) = \usf(svs')$, then $svs = svs'$.
\end{itemize}
As a result, the implication is demonstrated and the equivalence as well.
\end{proof}

From this point, we want to build the canonical operations on symbolic vector sets to preserve the canonicity, thereby eliminating redundancy in the representation and capitalising on optimisations resulting from its uniqueness.

\subsection{Canonical operations}

We begin by defining a union operation in such a way that elements are transferred from one set to the other, and such that each newly introduced element, as well as the target structure itself, is adjusted to preserve canonicity.
To aid in defining this new operation, we introduce several intermediate definitions that are required for it.

\begin{definition}[Remove empty symbolic vectors]
Let $svs \in \SVS$ be a symbolic vector set.
The function $rmEmpty: \SVS \rightarrow \SVS$ that removes the empty symbolic vectors from a symbolic vector set is defined as:
\begin{align*}
    rmEmpty(svs) = \set{sv \in svs ~|~ sv \not\in \varnothing_{sv}}
\end{align*}
\end{definition}

This function maintains the second condition of canonicity by removing symbolic vectors that could result in an empty set of functions.

In order to define the union, we need to define the canonical subtraction between two symbolic vectors.

\begin{definition}[Canonical difference of symbolic vectors]
\label{definition:canonical-ps:difference-ps}
Let $sv = (a,b) \in \SV, sv' = (c,d) \in \SV$ be canonical symbolic vectors.
The definition of canonical difference on symbolic vectors, noted $\csetminus_{sv}: \SV \times \SV \rightarrow \SVS$, is defined as:
\begin{align*}
& q_1, \dots, q_n \in Q \land q_1 \leq_f \dots \leq_f q_n \\
f&: \SV \times \powerset{Q} \rightarrow \SVS \\
f((x,y), \varnothing) &= \varnothing_{svs} \\
f((x,y), \set{q_1, \dots, q_n}) &= rmEmpty(\set{can_{sv}(x \cup \set{q_1}, y)}) \cup_{svs} f((x,y \cup \set{q_1}), \set{q_2, \dots, q_n}) \\ 
\\
sv \csetminus_{sv} sv' &= \begin{cases}
    \bm{if}~sv' \in \varnothing_{sv} & rmEmpty(\set{sv})\\
    \bm{else}& rmEmpty(\set{can_{sv}(a, b \cup c)} \cup_{svs} f((a \cup c,b), d))
\end{cases}
\end{align*}
\end{definition}

While the definition of difference appears succinct, a sub-function $f$ has been introduced to facilitate the recursive implementation. 
This sub-function conceptually aligns with the original difference operation. 
However, to prevent the generation of symbolic vectors that could undermine canonicity, this specialised function is employed.
Specifically, the difference operation produces distinct symbolic vectors for each function within $d$.
Moreover, there is a potential for these symbolic vectors to intersect with each other, potentially resulting in non-canonical forms.

\begin{example}
Let us observe an example of the application of the basic difference operation:
\begin{align*}
    (\set{(1,1)}, \varnothing) \setminus_{sv} &(\set{(3,3)}, \set{(5,8), (8,4)}) \\
    &= \\
    \set{(\set{(1,1)}, \set{(3,3)})&, (\set{(5,8)}, \varnothing), (\set{(8,4)}, \varnothing)}
\end{align*}
This version is clearly not canonical, because there exists an overlap between $(\set{(5,8)}, \varnothing)$ and $ (\set{(8,4)}, \varnothing)$.
\end{example}

To address this concern, the approach involves incorporating the preceding $q_d$, employed to construct its respective symbolic vector, as a constraint for generating subsequent symbolic vectors.
This strategy guarantees that previously added values are not reintroduced into subsequent symbolic vectors.
Furthermore, to respect canonicity, the sequence in which tuples within $d$ are used to form new symbolic vectors is regulated using the total order relation.
This mechanism ensures that the smallest symbolic vector encompasses the shared portion.

\begin{example}
Let us examine the previous example with the canonical difference:
\begin{align*}
    (5,8) \leq_f (8,4) & \\
    f((\set{(1,1)},& \varnothing), \set{(5,8), (8,4)}) \\
    =&~\set{can_{sv}(\set{(1,1)} \cup \set{(5,8)}, \varnothing)} \cup_{svs} f((\set{(1,1)}, \set{(5,8)}), \set{((8,4)}) \\
    =&~\set{(\set{(5,8)}, \varnothing)} \cup_{svs} f((\set{(1,1)}, \set{(5,8)}), \set{((8,4)}) \\
    =&~\set{(\set{(5,8)}, \varnothing)} \cup_{svs} \set{can_{sv}(\set{(8,4) \cup \set{(1,1)}}, \set{(5,8)})} \\
    &~\cup_{svs} f((\set{(1,1)}, \set{(5,8), (8,4)}), \varnothing) \\
    =&~\set{(\set{(5,8)}, \varnothing)} \cup_{svs} \set{(\set{(8,4)}, \set{(8,8)})} \cup_{svs} \varnothing_{svs} \\
    =&~ \set{(\set{(5,8)}, \varnothing), (\set{(8,4)}, \set{(8,8)})} \\
    \\
    (\set{(1,1)}, \varnothing)& \csetminus_{sv} (\set{(3,3)}, \set{(5,8), (8,4)}) \\
    &= \set{can_{sv}(\set{(1,1)}, \set{(3,3)})} \cup_{svs} f((\set{(1,1)}, \varnothing), \set{(5,8), (8,4)})\\
    &= \set{(\set{(1,1)}, \set{(3,3)})} \cup_{svs} \set{(\set{(5,8)}, \varnothing), (\set{(8,4)}, \set{(8,8)})} \\
    &= \set{(\set{(1,1)}, \set{(3,3)}), (\set{(5,8)}, \varnothing), (\set{(8,4)}, \set{(8,8)})} 
\end{align*}
This updated version exhibits no overlap and is inherently canonical. 
It is important to note that we have omitted the invocation of the function $rmEmpty$ here to enhance clarity.
Moreover, there are no symbolic vectors yielding an empty underlying set.
\end{example}

\begin{lemma}
\label{lemma:canonical-ps:canonical-difference-sv}
If $sv, sv' \in SV$ are two canonical symbolic vectors, then $sv \csetminus_{sv} sv'$ is canonical.
\end{lemma}

\begin{lemma}
\label{lemma:canonical-ps:merge-canonical}
If $sv, sv' \in SV$ are two canonical symbolic vectors and mergeable, then $merge_{sv}(sv,sv')$ is canonical.
\end{lemma}

We are now in a position to create the canonical union operation.

\begin{definition}[Canonical union]
\label{def:canonical-ps:union-svs}
Let $svs = \set{sv_1, \dots, sv_n}, svs' \in \SVS, n > 0$ be canonical symbolic vector sets.
The \textit{canonical union} of two symbolic vector sets, noted $\ccup_{svs}: \SVS \times \SVS \rightarrow \SVS$, is defined as:
\begin{align*}
\text{Given:}& \\
sv_1 &= (\set{q_a}, b),~ sv_1' = (\set{q_c}, d) \\
m &\in \Nat, m \geq 0 \\
svs '' &= \set{sv_1', \dots, sv_m'} = \set{sv' \in svs' ~|~ shareable_{sv}(sv_1, sv') = \true} \\
sh &= share_{sv}(sv_1, sv_1') \\
\\
\varnothing \ccup_{svs} svs &=  svs \\
\set{sv_1} \ccup_{svs} svs' &=
\begin{cases}    
\bm{if}&svs'' \neq \varnothing \\
    &\begin{cases}
        \bm{if}&q_a \leq_f q_c \\
        &(merge_{sv}(sv_1, sh) \cup_{svs} (sv_1' \csetminus_{sv} sh)) \ccup_{svs} (svs' \setminus \set{sv_1'}) \\
        \bm{if}& q_c \leq_f q_a\\
        &(merge_{sv}(sv_1', sh) \cup_{svs} (sv_1 \csetminus_{sv} sh)) \ccup_{svs} (svs' \setminus \set{sv_1'}) \\
    \end{cases} \\
    \bm{if}&svs'' = \varnothing\\
    &\set{sv_1} \cup_{svs} svs'
\end{cases} \\
\bm{if}~n > 1&: \set{sv_1, \dots, sv_n} \ccup_{svs} svs' = \set{sv_2, \dots, sv_n} \ccup_{svs} \set{sv_1} \ccup_{svs} svs' 
\end{align*}
\end{definition}

The fundamental concept underlying the union operation involves sequentially isolating each symbolic vector from the left-hand side.
Subsequently, all components of the other set are examined to identify those that include a shareable part with the symbolic vector that is isolated from it.
The set generated $svs''$ contains all of these symbolic vectors.
As we explained earlier, if shareability between two symbolic vectors is true, this means that a change is required to conform to canonicity.
A symbolic vector is selected in $svs''$ regardless of order.
The shared part between these two symbolic vectors is extracted and stored in $sh$, which is then merged with the smaller symbolic vector.
Given that the larger symbolic vector already contains this shared segment, a subtraction operation is applied to eliminate it.
Thus, this shared part is moved from one symbolic vector to the other with respect to the total order of the functions.
The same process is repeated as many times as necessary until there is only one set left.
When an element no longer possesses any shared components with the other set, it is included using a conventional set union operation.

\begin{remark}
Remark that the difference $svs \setminus \set{sv_1'}$ is a standard set difference operation.
We remove the $sv_1'$ element that was selected in $svs$.
\end{remark}

\begin{lemma}[Canonical union commutativity]
\label{lemma:canonical-ps:union-commutativity}
Let $svs, svs' \in \powerset{PS}$ be two symbolic vector sets.
Then,
\begin{align*}
svs \ccup_{svs} svs' = svs' \ccup_{svs} svs
\end{align*}
\end{lemma}

\begin{proposition}
\label{proposition:canonical-ps:canonical-union}
If $svs, svs' \in \SVS$ are two canonical symbolic vector sets, then $svs \ccup_{svs} svs'$ is canonical.
\end{proposition}

\begin{theorem}[Canonical union termination]
\label{thm:canonical-union-termination}
For all finite canonical symbolic vector sets
$svs,svs' \in \SVS$, the computation of
$svs \ccup_{svs} svs'$ terminates.
\end{theorem}

\begin{proof}
We proceed by analysing the different cases in
Definition~\ref{def:canonical-ps:union-svs}.

First, if the left-hand operand is empty, then
\[
    \varnothing \ccup_{svs} svs' = svs'.
\]
No recursive call is performed, and the computation therefore
terminates immediately.

Suppose next that the left-hand operand contains more than one
symbolic vector. For $n>1$, the definition gives
\[
\set{sv_1,\dots,sv_n}\ccup_{svs}svs'
=
\set{sv_2,\dots,sv_n}
\ccup_{svs}
\bigl(\set{sv_1}\ccup_{svs}svs'\bigr).
\]
This clause successively reduces the number of symbolic vectors in
the left-hand operand until the singleton case is reached. Since the
initial set is finite, this decomposition can be performed only
finitely many times. It is therefore sufficient to establish the
termination of
\[
    \set{sv_1}\ccup_{svs}svs'.
\]

Let
\[
svs'' =
\set{sv'\in svs' \mid shareable_{sv}(sv_1,sv')=\true}.
\]

If $svs''=\varnothing$, then
\[
    \set{sv_1}\ccup_{svs}svs'
    =
    \set{sv_1}\cup_{svs}svs'.
\]
This is a non-recursive set-union operation and hence terminates.

Assume now that $svs''\neq\varnothing$, and let
$sv_1'\in svs''$ be the selected symbolic vector. Moreover, let
\[
    sh = share_{sv}(sv_1,sv_1').
\]

Consider first the case $q_a\leq_f q_c$. The definition yields
\[
\bigl(
    merge_{sv}(sv_1,sh)
    \cup_{svs}
    (sv_1'\csetminus_{sv}sh)
\bigr)
\ccup_{svs}
\bigl(svs'\setminus\set{sv_1'}\bigr).
\]

The auxiliary operation $merge_{sv}(sv_1,sh)$ terminates and produces
at most one symbolic vector. Furthermore, by termination of the
canonical symbolic-vector difference,
\[
    sv_1'\csetminus_{sv}sh
\]
is computed after finitely many decomposition steps and produces a
finite symbolic vector set. Consequently,
\[
    merge_{sv}(sv_1,sh)
    \cup_{svs}
    (sv_1'\csetminus_{sv}sh)
\]
is finite.

The right-hand operand is also finite, since
\[
    svs'\setminus\set{sv_1'}
\]
is obtained by removing one element from the finite set $svs'$.

However, finiteness alone is not sufficient to establish recursive
termination. Progress follows from the treatment of the shared part
$sh$. By construction, $sh$ is incorporated into
$merge_{sv}(sv_1,sh)$ and removed from $sv_1'$ by
$sv_1'\csetminus_{sv}sh$. Thus, the conflict represented by $sh$ is
resolved by this recursive step. Moreover, neither $merge_{sv}$ nor
$\csetminus_{sv}$ reintroduces an already resolved shared part.

Let $\mu(svs,svs')$ denote the number of unresolved shareable parts
induced by the finite symbolic vectors occurring in the computation.
This is a natural number. In every non-trivial recursive step, at
least the selected shared part $sh$ is resolved and no resolved part
is reintroduced. Hence,
\[
    \mu_{\mathrm{next}} < \mu_{\mathrm{current}}.
\]
Since the strict order on $\Nat$ is well founded, only finitely many
such recursive calls can be performed.

The case $q_c\leq_f q_a$ is symmetric. In that case, the recursive
expression is
\[
\bigl(
    merge_{sv}(sv_1',sh)
    \cup_{svs}
    (sv_1\csetminus_{sv}sh)
\bigr)
\ccup_{svs}
\bigl(svs'\setminus\set{sv_1'}\bigr),
\]
and the same argument applies: the auxiliary operations terminate,
their results are finite, and the selected shared part $sh$ is
resolved without being reintroduced.

Therefore, every recursive branch either reaches a non-recursive base
case or strictly decreases the well-founded measure $\mu$. Hence,
$svs\ccup_{svs}svs'$ terminates for all finite canonical symbolic
vector sets $svs$ and $svs'$.
\end{proof}

\begin{definition}[Canonised SVS]
Let $svs = \set{sv_1, \dots, sv_n} \in \SVS, n > 0$ be a symbolic vector set.
The canonised function that returns the canonical form of a symbolic vector set, noted $can_{svs}$, is defined as follows:
\begin{align*}
can_{svs}(\varnothing) &= \varnothing \\
can_{svs}(\set{sv_1, \dots, sv_n}) =&~rmEmpty(\set{can_{sv}(sv_1)}) \\
&\ccup_{svs} \dots \ccup_{svs} rmEmpty(\set{can_{sv}(sv_n)})
\end{align*}
\end{definition}

This definition breaks down the set to reconstruct it using the canonical union.

\begin{proposition}
Let $svs \in \SVS$ be a symbolic vector set.
Then, $can_{svs}(svs)$ is canonical.
\end{proposition}

Thanks to the canonical union, all operators on symbolic vector sets can also be written canonically, which is the subject of the following subsections.

We define the canonical intersection of symbolic vector sets as follows.

\begin{definition}[Canonical intersection of symbolic vector sets]
Let $svs = \set{sv_1, \dots, sv_n},~svs' = \set{sv_1', \dots, sv_m'} \in \SVS, n,m > 0 $ be canonical symbolic vector sets.
The definition of canonical intersection, noted $\ccap_{svs}: \SVS \times \SVS \rightarrow \SVS$, is defined as:
\begin{align*}
svs \ccap_{svs} \varnothing &= \varnothing \\
\varnothing \ccap_{svs} svs &= \varnothing \\
svs \ccap_{svs} svs' &= \\
&\begin{cases}
\bm{if}& n = 1 \\
&rmEmpty(\set{sv_1 \ccap_{sv} sv_1'}) \ccup_{svs} (\set{sv_1} \ccap_{svs} 
\set{sv_2', \dots, sv_m'})\\
\bm{if}& n > 1 \\
&(\set{sv_1} \ccap_{svs} svs) \ccup_{svs} (\set{sv_2, \dots, sv_n} \ccap_{svs} svs)
\end{cases}
\end{align*}
\end{definition}

The original intersection in \cref{def:predicate-structure:union-intersection} is almost the same as the new definition.
There is a slight difference in dealing with the case of an empty intersection between symbolic vectors, where no function is added to the set.

\begin{proposition}
\label{proposition:canonical-ps:canonical-intersection}
If $svs, svs' \in \SVS$ are two canonical symbolic vector sets, then $svs \ccap_{svs} svs'$ is canonical.
\end{proposition}

The canonical difference between two symbolic vectors was defined earlier in~\cref{definition:canonical-ps:difference-ps}.
We now directly introduce the canonical difference operation for symbolic vector sets.

\begin{definition}[Canonical difference of symbolic vector sets]
Let $svs = (sv_1, \dots, sv_n), svs' = (sv_1', \dots, sv_m') \in \SVS$ be canonical symbolic vector sets.
The definition of canonical difference on symbolic vector sets, noted $\csetminus_{svs}: \SVS \times \SVS \rightarrow \SVS$, is defined as:  
\begin{align*}
svs \csetminus_{svs} \varnothing &= svs \\
\varnothing \csetminus_{svs} svs &= \varnothing \\
svs \csetminus_{svs} svs' &= \\
&\begin{cases}
\bm{if}& n = 1 \\
&(sv \csetminus_{sv} sv_1) \csetminus_{svs} \set{sv_2', \dots, sv_m'} \\
\bm{if}& n > 1 \\
&(\set{sv_1}  \csetminus_{svs} svs) \ccup_{svs} \set{sv_2, \dots, sv_n} \csetminus_{svs} svs
\end{cases}
\end{align*}
\end{definition}

The canonical difference is almost the same as the original~\cref{definition:predicate-structure:difference-svs}.
The sole distinction lies in the use of $\ccup_{svs}$ and $\csetminus_{sv}$ instead of their non-respective canonical form.

\begin{proposition}
If $svs, svs' \in \SVS$ are two canonical symbolic vector sets, then $svs \csetminus_{svs} svs'$ is canonical.
\end{proposition}

The last canonical operation to explore is the negation.

\begin{definition}[Canonical negation of a symbolic vector]
Let $sv \in SV$ be a symbolic vector.
The canonical \textit{negation} of a symbolic vector, noted $\cneg_{sv}: \SV \rightarrow \SVS$, is defined as:
\begin{align*}
\cneg_{sv}(sv) = (\set{f_\varepsilon}, \varnothing) \csetminus_{sv} sv
\end{align*}
\end{definition}

\begin{definition}[Canonical negation of a symbolic vector set]
Let $svs \in \SVS$ be a symbolic vector set.
The canonical \textit{negation} of a symbolic vector set, noted $\cneg_{svs}: \SVS \rightarrow \SVS$, is defined as:
\begin{align*}
\cneg_{svs}(svs) = \set{(\set{f_\varepsilon}, \varnothing)} \csetminus_{svs} svs
\end{align*}
\end{definition}

Both operations remain akin to their original versions, with the exception that we employ the canonical variant of the difference operation.

\begin{proposition}
If $sv \in SV$ is a canonical symbolic vector, then $\cneg_{sv}(sv)$ is canonical.
\end{proposition}

\begin{proposition}
If $svs \in \SVS$ is a canonical symbolic vector set, then $\cneg_{svs}(svs)$ is canonical.
\end{proposition}

To conclude, this section has presented the necessary algebraic framework for symbolic vectors with the required operations that can be used for global model checking of Petri nets.

\section{Petri net model checking using symbolic vector sets}
\label{section:from-svs-to-pn}

This section aims to describe the relation and application between CTL model checking and Petri nets using symbolic vector sets.

\subsection{Encoding markings using symbolic vector sets}

\begin{definition}[Capacity Petri net] 
A Petri net with capacities is a tuple $\Sigma = \tuple{P, T, in, out, w, k, m_0}$ where $P$ is the set of finite places, $T$ is the set of finite transitions, $\text{in} \subseteq (P \times T)$ and $\text{out} \subseteq (T \times P)$ are the sets of input and output arcs linking places and transitions, $w: (\text{in} \cup \text{out}) \rightarrow \mathbb{N}^+$ is the function labeling each arc with a weight, $k: P \rightarrow \mathbb{N}^+ \cup \set{\infty}$ binds each place with a value that cannot be exceeded, and $m_0 \in M$ represents the initial marking, belonging to the set of all markings $M$ and respecting the capacity. We define $M$ the set of all markings, $M=\{f|f:P \rightarrow \mathbb{N}\}$. The content of a place $p\in P$ in the marking $m\in M$ is noted m(p). This is similar to the notion of vector introduced in the previous chapter. We do not limit the markings to those that cannot exceeded for notation simplicity. The firing rule will assume that.
We assume that $w$ is total, returning $0$ when the relations $in$ or $out$ are not defined.
Formally, $w(p,t) = 0$ if $(p,t) \notin in$ and $w(t,p) = 0$ if $(t,p) \notin out$.
Furthermore, $\Nat^+ = \Nat \setminus \set{0}$.
Operators and relations on $\Nat$ are extended to $\Nat^+ \cup \set{\infty}$.
As previously defined a marking is denoted by a vector $(n_1, \dots, n_m)$ where $n_1, \dots, n_m \in \Nat, m = |P|$, in which each tuple location is connected to a place that indicates the number of resources. 
Vectors can be seen as functions of finite domain $m: P \rightarrow \Nat$.
A transition $t \in T$ is enabled for the marking $m \in M$ iff  $\forall p \in P, w(p,t) \leq m(p) \leq k(p) - w(t,p) + w(p,t) = isEnabled(m,t)$, where $\infty - n = \infty, n \in \Nat$ and $l \leq \infty$ is true for $l \in \Nat \cup \set{\infty}$.
When a transition is enabled, it may be fired, resulting in: $fire(m,t) = m' \Leftrightarrow \forall p \in P, m'(p) = m(p) - w(p,t) + w(t,p)$ such that $fire: M \times T \rightarrow M$.

We define a function, $\lambda_{in}: T \rightarrow M$, called input marking, such that $\forall p \in P, \lambda_{in}(t)(p) = w(p,t)$.
$\lambda_{in}$ returns a marking that contains the required number of tokens to enable a transition.
In the following, the definition $\Sigma = \tuple{P, T, in, out, w, k, m_0}$ is assumed for a Petri net.
\end{definition}

To apply CTL model checking on a model, a particular operation called 
\textit{pre} is necessary. This operation is used to calculate the execution trace of a model backwards, step by step~\cite{sifakis:1982:unified}.
In the context of a Petri net, the \textit{pre} operation serves as the inverse or opposite of the $fire$ operation.
This operation computes the set of predecessor markings for a given transition, representing the state space before the firing of the transition.
In our approach, encoding the operation $pre$ is crucial to construct the state space symbolically.
The choice of Petri nets with finite capacity ($k: P \rightarrow \mathbb{N}$) is also crucial to ensure the decidability of our proof system, which is not the case for general Petri nets.

\begin{definition}[Reversible and $pre$]
\label{definition:pn-as-ps:reversible-definition}
Let $t \in T$ be a transition and $m \in M$ be a marking.
The \textit{reversibility} of a transition $t$ for a marking $m$ is defined as $rev(m,t) \Leftrightarrow \forall p \in P,w(t,p) \leq m(p) \leq k(p) - w(p,t) + w(t,p)$.
In addition, the partial function that computes its application if it is reversible, noted $pre: M \times T \rightarrow M$, is defined as: $\forall p \in P, pre(m,t)(p) = 
    m(p) + w(p,t) - w(t,p) 
$
\end{definition}

$pre$ may be seen as the \textit{reverse fire} operation for a given transition i.e. exactly the fire of the transition with arcs in the reverse order.
Hence, the targeted transition is always fireable the next time allowing to explore all satisfiable states, related to global model checking.
Our objective is not to compute the reachable configurations from a given marking, but to capture all the initial configurations satisfying a property. 
Thus, even if a transition does not influence a place, it should still be considered in the list of all valid initial configurations, we will discuss that later in the model checking algorithms.
To operate on a set of markings, it involves applying the $pre$ operation to each marking with every transition and combining the results using the union operation.

\begin{figure}[ht!]
\centering
\begin{tikzpicture}[>=stealth']
\node[place, label={135:$p_0$}] at (-2,0) (p0) {4};
\node[place, label={45:$p_1$}] at (2,0) (p1) {$1$};
\node[place, label={45:$p_2$}] at (0,2) (p2) {$1$};
\node[transition, label={90:$t_0$}] at (0,0) (t0) {};

\draw[->] (p0) to node[above] {$2$} (t0);
\draw[->] (t0) to node[above] {$1$} (p1);
\draw[->, bend right=45] (p2) to node[left] {$1$} (t0);
\draw[->, bend right=45] (t0) to node[right] {$2$} (p2);
\end{tikzpicture}

\caption{A capacity-constrained Petri net with three places, capacity vector $(4,3,3)$, and one transition.}
\label{figure:modelchecking:petri-net-pre}
\end{figure}

Given the Petri net in \cref{figure:modelchecking:petri-net-pre}, consider the initial marking
$$
m_0=(4,1,1).
$$

Firing \(t_0\) from \(m_0\) yields
$$
\operatorname{fire}(m_0,t_0)=(2,2,2).
$$
Consequently,
$$
\operatorname{rev}((2,2,2),t_0)
$$
holds and
$$
\operatorname{pre}((2,2,2),t_0)=(4,1,1).
$$
The minimal marking enabling \(t_0\), denoted by \(m_{\min}\), is the component-wise smallest marking from which \(t_0\) can be fired. In this example,
$$
m_{\min}=(2,0,1).
$$
Thus, in the absence of finite place capacities, every marking in $\{m \in M~|~m_{\min} \subseteq_f m\}$ enables \(t_0\), where the ordering on markings is understood component-wise.

When places have finite capacities, satisfying $m_{\min} \subseteq_f m$ is no longer sufficient. The firing must also produce a marking that respects the capacity of every place. More precisely, 
$$
t_0~\text{is enabled at}~m
\Leftrightarrow
\forall p\in P,~
w(p,t_0)
\leq m(p)
\leq
k(p)-w(t_0,p)+w(p,t_0).
$$

The first inequality ensures that sufficiently many tokens are available, whereas the second ensures that firing the transition does not cause the capacity of \(p\) to be exceeded.

The lower enabling condition can therefore be represented by \(m_{\min}\), while the capacity constraints can be represented by a set of upper-bound vectors. Let
$$
b_{cap}\in\powerset{M}.
$$
where \(b_{\mathrm{cap}}\) denotes the second component of a symbolic vector \((a,b_{\mathrm{cap}})\in\SV\). The subscript \(\mathrm{cap}\) is used solely to indicate that this component represents a set of capacity-related markings. For each place \(p\in P\) with finite capacity, define the marking \(\pi_p\) as follows:

$$
\pi_p(x)=
\begin{cases}
k(p)+1 & \text{if }x=p,\\
0      & \text{otherwise}.
\end{cases}
$$

The set of markings  representing all place-capacity constraints is then

$$
b_{cap}
=
\left\{
\pi_p
\;\middle|\;
p\in P,\ k(p) \neq \infty
\right\}.
$$

For the net in \cref{figure:modelchecking:petri-net-pre}, whose capacity vector is \((4,3,3)\), the set of upper-bound constraints is

$$
b_{cap}
=
\left\{
(5,0,0),(0,4,0),(0,0,4)
\right\}.
$$

Each vector represents the violation of the capacity of one particular place.

It is essential to keep these constraints separate. In particular,
$$
uf_{svs}((\{(0,0,0)\}, \{(5,4,4)\}))
\neq
uf_{svs}((\{(0,0,0)\}, \{(5,0,0),(0,4,0),(0,0,4)\}))
$$

For example,
$$
(5,2,2) \in uf_{svs}((\{(0,0,0)\}, \{(5,4,4)\})) \\
$$
$$
(5,2,2) \not\in uf_{svs}((\{(0,0,0)\}, \{(5,0,0),(0,4,0),(0,0,4)\}))
$$

Thus, the symbolic vector $((\{(0,0,0)\}, \{(5,4,4)\}))$ is not sufficient to ensure all the capacity constraints.
 
\begin{definition}[Operation $pre_{sv}$]
\label{definition:pn-as-ps:pre-ps}
Let $sv = (a,b) \in \SV$ be a symbolic vector.
The operation $pre_{sv}: \SV \rightarrow \SVS$ is defined as:
\begin{align*}
    pre_{sv}(a,b) &= \bigcup_{t \in T} pre_{svt}((a,b),t) \\
    pre_{svt}&: \SV \times T \rightarrow \SVS \\
    pre_{svt}((a,b),t) &=
    \begin{cases}
    \varnothing~\bm{if}~\exists m \in (a \cup b), \neg rev(m,t) \\
    \{can_{sv}((pre_i(a,t), pre_e(b,t))\}~\bm{else}
    \end{cases} \\
    pre_i, pre_e&: \powerset{M} \times T \rightarrow \powerset{M} \\
    pre_i(\varnothing, t) &= \set{\lambda_{in}(t)},~pre_e(\varnothing, t) = \varnothing \\
    pre_i(a,t) &= pre_e(a,t) = \bigcup_{m~\in~a} \set{pre(m,t)}, ~\bm{if}~a \neq \varnothing
\end{align*}
\end{definition}
The operation $pre_{sv}$ is designed to operate on each marking within the symbolic vector individually, computing its corresponding $pre$. It is based on a more specific function $pre_{svt}$ doing this transition by transition.
Note that if one of the markings is not reversible ($\neg rev(m,t)$), the entire operation is not applied and returns $\varnothing$.
We do not need to be more precise as we are doing global model checking, and only symbolic vectors with same behavior will be considered.
Moreover, it preserves the capacity of Petri nets with this constraint. Checking reversibility can be done without enumerating explicitly markings but checking the bounds as in $\exists m \in (a \cup b)$.

Let us look at another examples from Petri net in figure \ref{figure:modelchecking:petri-net-pre}:
\begin{align*}
    pre_{sv}((\{(0,0,0)\}, \{(5,4,4)\})) &= \bigcup_{t~\in~\{t_0\}} pre_{svt}((\{(0,0,0)\}, \{(5,4,4)\}),t) \\
    &= pre_{svt}((\{(0,0,0)\}, \{(5,4,4)\}),t_0) \\
    &= \varnothing_{svs} \\
    \\
    rev((5,4,4), t_0) &\Leftrightarrow \forall p \in P,w(t_0,p) \leq m(p) \leq k(p) - w(p,t_0) + w(t_0,p)\\
    \text{if}~(p = p_0),&~w(t_0,p_0) \leq m(p_0) \leq k(p_0) - w(p_0,t_0) + w(t_0,p_0) \\
    &\Rightarrow 0 \leq 4 \leq 4 - 2 + 1 \Rightarrow 0 \leq 4 \leq 3 \\
     \neg rev((5,4,4), t_0) &= \neg false = true \\
\end{align*}

Because there exists a marking $m = (5,4,4)$ such that it is not reversible for $t_0$, the final result is $\varnothing_{svs}$.

\begin{align*}
    pre_{sv}((\{(0,1,2)\}, \{(2,2,3)\})) &= \bigcup_{t~\in~\{t_0\}} pre_{svt}((\{(0,1,2)\}, \{(2,2,3)\}),t) \\
    &= pre_{svt}((\{(0,1,2)\}, \{(2,2,3)\}),t_0) \\
    &= \{(\{2,0,1\},\{(4,1,2)\})\}\\
    \\
    rev((2,2,3), t_0) &\Leftrightarrow \forall p \in P,w(t_0,p) \leq m(p) \leq k(p) - w(p,t_0) + w(t_0,p) \\
    \text{if}~(p = p_0), &\Rightarrow 0 \leq 2 \leq 4 - 2 + 1 \Rightarrow 0 \leq 2 \leq 3 \\
    \text{if}~(p = p_1), &\Rightarrow 1 \leq 2 \leq 3 - 0 + 1 \Rightarrow 1 \leq 2 \leq 4 \\
    \text{if}~(p = p_2), &\Rightarrow 2 \leq 2 \leq 3 - 1 + 2 \Rightarrow 0 \leq 2 \leq 4 \\
     rev((2,2,3), t_0) &=  true \\
     rev((0,1,2), t_0) &=  true \\
\end{align*}

in this case the reversibility is possible within the range of markings $(\{(0,1,2)\}, \{(2,2,3)\})$.


\begin{definition}[Operation $pre_{svs}$]
\label{def:pn-as-ps:pre-svs}
Let $svs \in \SVS$.
Operation $pre_{svs}: \SVS  \rightarrow \SVS$ is defined as: $pre_{svs}(svs) = \bigcup_{sv~\in~svs}~pre_{sv}(sv)$

\end{definition}

\begin{lemma}(Validity of the symbolic $pre$)
Let  $sv \in \SV$ be a symbolic vector. 
Then the $pre$ of each marking of a symbolic vector produces at least the results as the $pre$ on the symbolic vector $ uf(pre_{sv}(sv)) \subseteq \bigcup_{t\in T}\{pre(m,t)\mid m\in uf(sv)\} $ and similarly for its extension to set of symbolic vectors $svs  \in \SVS$.
\label{lemma:homomorphicsvs}
\end{lemma}

Validity is here sufficient as we will iterate exactly on the interval of markings compatible with the firing of transitions, i.e. we have a partial completness conditional to the reversibility of the whole interval and defined as $\forall t \in T$ if ($\forall m \in (a,b),rev(m,t)) \Rightarrow uf(pre_{sv}(sv)) = \{pre(m,t)\mid m\in uf(sv)$
\begin{remark}
    Our approach is adaptable to models beyond classical Petri nets.
    For models other than Petri nets, encoding "\textit{pre}" is necessary to apply model checking for CTL and the fact that lemma \ref{lemma:homomorphicsvs} is true. In our case it is true as monotonocity is a valid property (modulo capacity). 
    It remains an open question whether this is always possible with symbolic vector sets.
    In some cases, it can be necessary to have an infinite number of symbolic vectors, such as a Petri net representing even or odd numbers of tokens.
\end{remark}

It must be noted that in global model checking many techniques including the one used in the paper, do not require explicit construction of reachability sets.  In our case, depending on the choice of initial marking any marking can be reached, this is why it is more pertinent to use CTL global model checking, i.e. the states that satisfy a given formula.

The three Petri nets presented in~\cref{figure:modelchecking:petri-nets} illustrate different situations regarding the finiteness of their sets of reachable markings and their representation using symbolic vector sets. Each net contains a single place $p_0$ with capacity $\infty$ and a single transition $t_0$. They differ only in the weight of the output arc from $t_0$ to $p_0$.

\begin{figure}[ht!]
\centering

\begin{subfigure}[b]{0.31\textwidth}
\centering
\begin{tikzpicture}[>=stealth']
\node[place, label={135:$p_0$}] at (-1,0) (p0) {1};
\node[transition, label={90:$t_0$}] at (1,0) (t0) {};

\draw[->,bend right=45] (p0) to node[below] {$1$} (t0);
\draw[->,bend right=45] (t0) to node[above] {$1$} (p0);
\end{tikzpicture}

\caption{One token produced}
\label{figure:modelchecking:petri-net-1}
\end{subfigure}
\hfill
\begin{subfigure}[b]{0.31\textwidth}
\centering
\begin{tikzpicture}[>=stealth']
\node[place, label={135:$p_0$}] at (-1,0) (p0) {1};
\node[transition, label={90:$t_0$}] at (1,0) (t0) {};

\draw[->,bend right=45] (p0) to node[below] {$1$} (t0);
\draw[->,bend right=45] (t0) to node[above] {$2$} (p0);
\end{tikzpicture}

\caption{Two tokens produced}
\label{figure:modelchecking:petri-net-2}
\end{subfigure}
\hfill
\begin{subfigure}[b]{0.31\textwidth}
\centering
\begin{tikzpicture}[>=stealth']
\node[place, label={135:$p_0$}] at (-1,0) (p0) {1};
\node[transition, label={90:$t_0$}] at (1,0) (t0) {};

\draw[->,bend right=45] (p0) to node[below] {$1$} (t0);
\draw[->,bend right=45] (t0) to node[above] {$3$} (p0);
\end{tikzpicture}

\caption{Three tokens produced}
\label{figure:modelchecking:petri-net-3}
\end{subfigure}

\caption{Petri nets with one place of capacity $\infty$ and one transition, differing only in the weight of the output arc.}
\label{figure:modelchecking:petri-nets}
\end{figure}

In~\cref{figure:modelchecking:petri-net-1}, firing $t_0$ consumes one token from $p_0$ and produces one token in the same place. The marking therefore remains unchanged after each firing, and the set of reachable markings is $\{(1)\}$ .

This finite set can be represented by the symbolic vector $(\{(1)\},\{(2)\})$.

In~\cref{figure:modelchecking:petri-net-2}, firing $t_0$ consumes one token and produces two tokens in $p_0$. Each firing therefore increases the marking by one, resulting in the infinite set of reachable markings $\{(1),(2),(3),\ldots\}$.
Despite containing infinitely many markings, this set can be represented by the single symbolic vector $(\{(1)\},\emptyset)$.
This example shows that our symbolic representation can finitely represent an infinite set of reachable markings.

Finally, in~\cref{figure:modelchecking:petri-net-3}, firing $t_0$ consumes one token and produces three tokens in $p_0$. Each firing therefore increases the marking by two. Starting from $(1)$, the place can only contain an odd number of tokens, and the set of reachable markings is
$$
\{(1),(3),(5),(7),\ldots\}.
$$
This set cannot be represented by a finite symbolic vector set. Indeed, a separate symbolic vector is required for each reachable marking:
$$
\{(\{(1)\},\{(2)\}),(\{(3)\},\{(4)\}),(\{(5)\},\{(6)\}),\ldots\}.
$$ 
These examples illustrate both the expressiveness and the limitations of our approach. A symbolic vector set can provide a finite representation of an infinite set of markings, as shown in~\cref{figure:modelchecking:petri-net-2}. However, this is not always possible: some infinite sets of markings, such as the set of odd markings generated by the Petri net in~\cref{figure:modelchecking:petri-net-3}, require an infinite number of symbolic vectors.

\subsection{Global model checking of Petri nets with symbolic vector sets}

In tackling the global model checking problem, our focus has been on addressing CTL formulas.
We briefly introduce the corresponding formalism and the satisfaction relation.



\begin{definition}[CTL definition and satisfaction]
Let \textit{AP} be a set of atomic propositions and $ap \in AP$.
The set of CTL formulas over \textit{AP}, noted $\phi \in CTL$, is defined as:
$\phi ::= ap~|~\neg \phi~|~\phi \lor \phi~|~\mathbf{EX}~\phi~|~\mathbf{EG}~\phi~|~\mathbf{E}[\phi~\mathbf{U}~\phi]$.
We assume the usual CTL extended syntax based on formula equivalence: $\mathbf{AX}~\phi$, $\mathbf{EF}~\phi$, $\mathbf{AF}~\phi$, $\mathbf{AG}~\phi$, and $\mathbf{A}[\phi~\mathbf{U}~\phi]$.

Let $\Sigma = \tuple{P, T, in, out, w, k, m_0}$ be a Petri net, $m \in M$ be a marking and $t \in T$ be a transition. We define the set of markings and the runs that are path within the marking graph : $ \rho : \mathbb{N}\rightarrow M$. They are defined as: $$\{\rho : m \in M\mid \forall i \in\mathbb{N}, if \rho(i) = m, \exists t \in T, fire(t,m)=m' \hbox{  \ then }\rho(i+1)=m'\}$$

As we consider with this definition only infinite path, we artificially put loops when deadlocks are reached, this is simply the condition for $\not\exists t\in T$ s.t. $fire(t,m)=m'$, in this case if $\rho(i) = m$ then $\rho(i + 1) = m$.

The \textit{satisfaction relation} is denoted as $\models \, \subseteq \Sigma \times M \times CTL$ and its definition based on the standard literature~\cite{clarke:1981:design} ($\Sigma$ is omitted in the final notation if the context is clear).
Furthermore, the set of atomic propositions for our Petri nets is $AP = \set{isFireable(t)~|~t \in T}$ such that:
$m \models isFireable(t) \Leftrightarrow$ t is enabled for $m $.

\begin{align*}
m \models \neg \phi \Leftrightarrow \neg (m \models \phi) \\ \; \\
m \models  true \\ \; \\
m \models \phi \lor \psi \Leftrightarrow m \models \phi  \ or \  m \models \psi \\ \; \\
m \models isFireable(t) \Leftrightarrow 
\forall p \in P, w(p,t) \leq m(p) \leq k(p) - w(t,p) + w(p,t)
\\ \; \\
m \models \mathbf{EX}~\phi \Leftrightarrow \exists m' \in M,t \in T,  m'=fire(m,t) \ and \  m'\models \phi \\ \; \\
    m \models \mathbf{EG} \phi_1  \Leftrightarrow \hbox{there is a run\ }\rho\hbox{\ with\ }\rho(0)=m \\
   \ \hbox{\ and\ }\rho(i)\models \phi_1 \hbox{\ for all\ } i\ge 0\,\\\; \\
   m \models\phi_1\mathbf{EU}\phi_2  \Leftrightarrow  \hbox{there is a run\ }\rho\hbox{\ with\ }\rho(0)=m
  \hbox{\ and\ } k\ge0 \hbox{\ s.t.\ }\\
   \ \rho(i)\models\phi_1 \hbox{\ for all\ } i<k
	\hbox{\ and\ }
	\rho(k)
	\models\phi_2
	\,\\
\end{align*}

\end{definition}

The atomic proposition $isFireable(t)$ is intended to observe the firing of transitions in Petri nets. It is obviously related to a particular state of the net i.e. the pre/postcondition of that transition. For $\mathbf{EG}$ and $\mathbf{EU}$ operators we have two cases to cover the finite or infinite paths cases reflected by the or operators in the formula.
We have assembled all the necessary components to provide the evaluation semantics of CTL formulas as the computation of a symbolic vector set.
This semantics has been revisited on the basis of~\cite{racloz:1994:properties}.

\begin{definition}[Evaluation of a CTL formula]
Let $\phi, \psi \in CTL$, $ap \in \set{isFireable(t)~|~t\in T}$, and $\Sigma = \tuple{P, T, pre, post, w, k, m_0}$ be a Petri net.
The \textit{evaluation} of a CTL formula as a symbolic vector set is noted $\eval{~}: CTL \rightarrow \SVS$.
The function is defined as follows:
\begin{align*}
\widetilde{pre}_{svs}(\eval{\phi}) &= \neg_{svs}pre_{svs}(\neg_{svs}\eval{\phi}) ~~~~~~~~~ \eval{\neg \phi} = \neg_{svs} \eval{\phi} \\
\eval{true} &= \set{(\set{f_\varepsilon}, \varnothing)} ~~~~~~~~~~~~~~~~ \eval{\phi \lor \psi} = \eval{\phi} \cup_{svs} \eval{\psi} \\
\eval{isFireable(t)} &= \set{(\set{\lambda_{in}(t)}, \varnothing)} ~~~~~~~~~~ \eval{\mathbf{EX}~\phi} = pre_{svs}(\eval{\phi}) \\
\eval{\mathbf{EG}~\phi} &= \nu Y . \eval{\phi} \cap_{svs} (pre_{svs}(Y) \cup_{svs} \widetilde{pre}_{svs}(Y)) \\
\eval{\mathbf{E} [\phi~\mathbf{U}~\psi]} &= \mu Y . \eval{\psi} \cup_{svs} (\eval{\phi} \cap_{svs} pre_{svs}(Y))
\end{align*}
\end{definition}

\begin{remark}
The semantics of satisfaction and evaluation both closely resemble those of a Kripke structure.
However, due to the left-total relation on Kripke structure arcs (i.e., each state has at least one arc starting from it), sink states are nonexistent. 
The semantics need to be adapted to account for this, expanding its scope, assuming an implicit continuation for sink states. 
For instance, the evaluation of $\mathbf{EG}$ in a Kripke structure could be simplified as follows: $\nu Y . \eval{\phi} \cap pre_{svs}(Y)$, which does not require the computation of $\widetilde{pre}_{svs}(Y)$ and the union application.
Furthermore, a Petri net without sink states could also benefit from this advantage. Termination is not guaranteed, in the general unbounded case, for the smallest fixpoints ($\mu Y$) or the largest fixpoint ($\nu Y$).
\end{remark}

\begin{theorem}
Let $\phi \in CTL$ a CTL formula and $m \in M$ a marking.
Then, $m \models  \phi \Leftrightarrow m \in_{svs} \eval{\phi}$.
\end{theorem}

The theorem asserts that the satisfaction of a CTL formula for a given marking can be determined by evaluating the formula and verifying whether the marking belongs to it.
This feature enables us to perform exhaustive computations to identify all states that satisfy a given formula.
The perspective of computing $\eval{\phi}$ aligns with the challenge of \textit{global model checking}, where the focus extends beyond a single configuration. Its validity is based on the previous lemma \ref{lemma:homomorphicsvs} and the well known result of model checking where we tuned the particular case of finite length computations in the path and the blocking in the computation of $pre_{svs}$.

This approach guides the remainder of the article.

\subsection{Use case: Mutual exclusion}

\begin{figure}[ht]
\centering
\scalebox{0.7}{
    \begin{tikzpicture}[>=stealth']
    \node[place, label={180:$p_0$}] at (0,0) (p0) {};
    \node[place, label={180:$p_1$}] at (0,-2) (p1) {};
    \node[place, label={135:$p_3$}] at (3, 0) (p3) {};
    \node[place, label={-135:$p_4$}] at (3, -2) (p4) {};
    \node[place, label={180:$p_2$}] at (3, -1) (p2) {};
    \node[transition, label={135:$t_0$}] at (1.5, 0) (t0) {};
    \node[transition, label={0:$t_3$}] at (4.5, 0) (t3) {};
    \node[transition, label={-135:$t_1$}] at (1.5, -2) (t1) {};
    \node[transition, label={0:$t_4$}] at (4.5, -2) (t4) {};
    
    \draw[->] (p0) to node[above] {} (t0);
    \draw[->] (p1) to node[above] {} (t1);
    \draw[->] (t0) to node[above] {} (p3);
    \draw[->] (t1) to node[above] {} (p4);
    \draw[->] (p2) to node[above, xshift=1mm] {} (t0);
    \draw[->] (p2) to node[above] {} (t1);
    \draw[->] (p3) to node[above] {} (t3);
    \draw[->] (p4) to node[above] {} (t4);
    \draw[->] (t3) to node[below, xshift=1mm] {} (p2);
    \draw[->] (t4) to node[above] {} (p2);
    \draw[->, bend right=45] (t3) to node[below] {} (p0.north);
    \draw[->, bend left=45] (t4) to node[below] {} (p1.south);
    \matrix [draw] at (7,-1) {
      \node [] {$k(p_0) = \infty$}; \\
      \node [] {$k(p_1) = \infty$}; \\
      \node [] {$k(p_2) = \infty$}; \\
      \node [] {$k(p_3) = \infty$}; \\
      \node [] {$k(p_4) = \infty$}; \\
    };
    \end{tikzpicture}
    }
    \caption{Petri net modelling the mutual exclusion problem. 
    }
    \label{figure:pn-as-ps:use-case-pn2}
\end{figure}

\Cref{figure:pn-as-ps:use-case-pn2} illustrates a Petri net example containing the well-known \textit{mutual exclusion problem}.
Markings are displayed using the order: $(p_0, p_1, p_2, p_3, p_4)$.
Avoiding mutual exclusion in this Petri net can be expressed by the property: $\forall m \in M, m \models \neg (\mathbf{EF}~isFireable(t_3) \land isFireable(t_4))$, computed as follows:
\begin{align*}
    \eval{\mathbf{EF}~isFire&able(t_3) \land isFireable(t_4)} = \mu Y . \set{(\set{(0,0,0,1,1)}, \varnothing)} \cup_{svs} pre_{svs}(Y) \\
    (1) &= \set{(\set{(0,0,0,1,1)}, \varnothing)}\\
    (2) &= \set{(\set{(0,0,0,1,1)}, \varnothing), (\set{(0,1,1,1,0)}, \varnothing), (\set{(1,0,1,0,1)}, \varnothing)} \\
    (3) &= \set{(\set{(0,0,0,1,1)}, \varnothing), (\set{(0,1,1,1,0)}, \varnothing), (\set{(1,0,1,0,1)}, \varnothing), \\
    &(\set{(1,1,2,0,0)}, \varnothing), (\set{(0,1,0,2,0)}, \varnothing), (\set{(1,0,0,0,2)}, \varnothing)}
\end{align*}
In the fixpoint computation of $\mathbf{EF}$, we have omitted the detailed steps and provided the final result of each step.
The result does not change after step (3).
Hence, a marking that is part of $\eval{\mathbf{EF}~isFireable(t_3) \land isFireable(t_4)}$ is considered not safe from the mutual exclusion problem.
For example, $(1,1,1,0,0)$ is a valid marking, whereas $(1,1,2,0,0)$ is not.
This result also yields an infinite number of valid markings.
Note that when the capacity of places is unbounded, the computation may not terminate if the solutions diverge.
Conversely, if all places are bounded, the computation always concludes.

\subsection{Optimisation through saturation}

Similarly to methods such as decision diagrams that encounter a \textit{peak}~\cite{ciardo:2001:saturation,hamez:2008:hierarchical} effect, the construction of symbolic vector sets is subject to the same challenge.
While computing the solution iteratively through the fixed-point, intermediate solutions are generated.
However, among these steps, some of the intermediate solutions will converge to the same solution. 
Thus, we end up working with more objects than necessary to obtain the final solution.

To address this issue, we propose a solution called \textit{saturation}, which effectively mitigates the \textit{peak} effect during computation.

\begin{figure}[ht!]
\centering
\begin{lstlisting}[style=swiftCode]
func evalEF($\phi$: CTL) -> SVS {
    n = 1 // The current capacity
    svsRes = $\eval{\phi}_{svs}$
    weights = collectLabelWeights().sorted()
    while n $\leq$ k {
        svsCap = svsRes // Keep result from previous capacity
        do {
            svsTemp = svsRes
            svsRes = svsRes $\cup_{svs}~pre_{svs_n}(svsRes)$ 
        } while !(svsRes $\subseteq_{svs}$ svsTemp)
        if svsRes $\subseteq_{svs}$ svsCap {
            weight = weights.first(where: {w > n})
            if weight.isNil() { break }
            n = weight
        } else { n = n + 1 } // Increment current capacity
    }
    return svsRes
}
\end{lstlisting}    
\caption{Saturated algorithm of the evaluation of $\mathbf{EF}~\phi$}
\label{fig:on-as-ps:saturated-ef-algo}
\end{figure}

\Cref{fig:on-as-ps:saturated-ef-algo} is the pseudo-code summarising this notion for the CTL temporal operator $\mathbf{EF}$.
$svsRes$ is the symbolic vector set storing the solution that will evolve, initialised with the previous evaluation of $\phi$ (line 3).
$weights$ is an array that compiles all weight labels from a Petri net, sorted from the smallest to the largest (line 4).
The functions \textit{collectLabelWeights} and \textit{sorted} are given.

To streamline our approach, we consider the capacity $k$ as a uniform bound for all places.
The idea is to iterate on $k$ (using $n$) starting from 1 (line 2) and compute the intermediate results of the CTL formula for a lower $k$. 
Subsequently, the result of the previous step is leveraged to compute the next one, creating an iterative process (from line 5 to 16). This process is in a straightforward way valid as the semantics of the computations is not modified, it is also complete as we keep the idea of fixpoint for the maximal capacity.

From line $7$ to $10$, we have the computation of fixed-point for $\mathbf{EF}$.
The main difference from the common evaluation is the use of $pre_{svs_n}$.
This function defined earlier in~\Cref{def:pn-as-ps:pre-svs} has been complemented by the index $n$.
This index specifies the current capacity used when computing $pre_{svs}$.
It ensures that symbolic vectors cannot contain markings with a value greater than $n$.

Upon completing the computation of the fixed-point for a given capacity, we compare the new result $svsRes$ with the previous one stored in $svsCap$ (line 6).
Furthermore, if the result of two consecutive iterations for two different capacities remains unchanged (line 11), we inspect the weights collected previously (line 4) to determine whether any of them is capable of generating a marking greater than the current capacity $n$ but less than $k$ (line 12).
The function $first$ serves as a filter that attempts to extract a weight from the array $weights$ such that its value is greater than $n$.
Here, $w$ is the variable that iterates through each value of $weights$.
If the function fails to find such a value, it returns $nil$. 
This means that no transition in the net is capable of altering the final result. 
Then, the $break$ statement (line 13) is executed to exit the loop and return $svsRes$ (line 17).
If such a weight is found, we update the corresponding capacity (line 14), allowing the potential discovery of new symbolic vectors in the next iteration.
On the other hand, if condition line 11 does not hold, $n$ is incremented by one (line 15).
Note that if the value of $weight$ is greater than $k$, the condition of the loop $while$ will also terminate.
For example, consider an arc of a transition that requires 6 tokens with $k=9$.
If $svsRes$ remains unchanged for $n = 2$ and $n = 3$, and the only other transition available also demands 6 tokens for an arc, we can update $n$ to $6$.
In fact, if no new transition is available under such conditions, the creation of new markings is not possible.

To make this pseudo-code work for other CTL temporal operators, it is enough to update line 9 for the corresponding CTL operation.
In addition, line 10 and 11 should be modified for the operations $\mathbf{EG}$ and $\mathbf{AG}$ due to the greatest fixed-point, requiring the elements of both relations to be swapped.



\section{Benchmarks and results}
\label{section:benchmark}

This section is dedicated to examining the results obtained with our tool, namely \textit{SVSKit}, which is a Swift library available on~\href{https://github.com/damdamo/SVSKit}{GitHub}.
Our library implements the theory developed in this article.
All the showcased examples are sourced directly from the \href{https://mcc.lip6.fr/2022/}{model checking contest} for the year 2022.
However, our tool is presented from the perspective of global model checking. 
All tests were conducted on a MacBook Pro computer equipped with a 3.2 GHz CPU and 32 GB of RAM.

\subsection{Circadian clock: A general model}
\label{subsection:benchmark:circadian-clock}

The Petri net of the \href{https://mcc.lip6.fr/2023/pdf/CircadianClock-form.pdf}{circadian clock} model, originally defined in~\cite{schwarick:2009:csl}, consists of 14 places, 16 transitions, and 58 arcs.
Several places in this net are associated with a variable value $N$, which can be adjusted according to specific objectives.
In our system, $N$ can be considered as the number that determines the capacity of each place.
In addition, increasing this value raises the system's complexity.

\begin{table}[ht]
    \centering
    \begin{tabular}{|c|c|c|c|c|}
    \hline
    Capacity & Saturated & Time \textbf{(s)} & Peak SV number & Final SV number\\
    \hline
    \multirow{2}{*}{1} & \cellcolor{lightgray}$true$ & \cellcolor{lightgray}1.9 & \cellcolor{lightgray}29 & \multirow{6}{*}{26} \\
    & $false$ & 1.9 & 29 & \\
    \cline{1-4}
    \multirow{2}{*}{2} & \cellcolor{lightgray}$true$ & \cellcolor{lightgray} 3.2 & \cellcolor{lightgray}29 & \\
    & $false$ & 45 & 86 & \\
    \cline{1-4}
    \multirow{2}{*}{6} & \cellcolor{lightgray}$true$ & \cellcolor{lightgray} 3.2 & \cellcolor{lightgray}29 & \\
    & $false$ & 19000 & 954 & \\
    \hline  
    \end{tabular}
    \caption{Comparison of a reachability property computation ($\mathbf{EF~\phi}$) for the \href{https://mcc.lip6.fr/2022/pdf/CircadianClock-form.pdf}{circadian clock} model, with and without the use of saturation.}
    \label{tabular:benchmark:circadian-clock-saturated-ef}
\end{table}

In \Cref{tabular:benchmark:circadian-clock-saturated-ef}, we compare the calculation of a reachability property both with and without saturation, for various capacities.
Saturation optimisation is instrumental in making symbolic vector sets practical for Petri nets of reasonable size.
It is worth noting that without saturation, the time complexity grows exponentially, reaching over 5 hours for a capacity of 6.
In contrast, with saturation enabled, the time remains constant starting from capacity 2 and above.
The use of saturation allows us to build symbolic vector sets step by step with the minimum of information required.
This can be observed in the \textit{Peak SV number} column, where the saturated solution remains consistently below 29 and does not exhibit significant growth.
Moreover, the final solution consistently comprises 26 symbolic vectors for all capacities, indicating that increasing it should not significantly alter the computational complexity.

We have implemented the \textit{query reduction} optimisation discussed in~\cite{bonneland:2018:simplification}.
In fact, because the CTL formulas for the competition are generated randomly, we are interested in simplifying them before evaluating them.
 However, only some of the reduction rules are available in the context of global model checking.
Furthermore, only the reduction rules concerning global model checking have been reused, 
which does not require any initial marking i.e. in such a way that it is not necessary to give any initial marking.


\begin{table}[ht!]
    \centering
    \begin{tabular}{|c|c|c|c|}
    \hline
    Capacity  & Marking number $(\sim)$ & Final SV number & Time \textbf{(s)}\\
    \hline
    100 & $1.1 \times 10^{28}$ & \multirow{4}{*}{26} & \multirow{4}{*}{3.2} \\
    \cellcolor{lightgray} 10000 & \cellcolor{lightgray}$1.0 \times 10^{56}$ & & \\
    1000000 & $1.0 \times 10^{84}$ & &\\
    \cellcolor{lightgray} $\infty$ & \cellcolor{lightgray} $\infty$ & &\\
    \hline  
    \end{tabular}
    \caption{Progression of the prior $\mathbf{EF}$ reachable property from~\Cref{tabular:benchmark:circadian-clock-saturated-ef} across varying capacities, assuming the saturation by default.
    }
    \label{tabular:benchmark:circadian-clock-saturated-marking-number}
\end{table}

In \Cref{tabular:benchmark:circadian-clock-saturated-marking-number}, we observe the behaviour when we assume saturation by default and only vary the capacity.
This result aligns with the previous finding, indicating that changing the capacity no longer significantly impacts the efficiency of the computation.
The primary result affected by this change is the number of markings encoded by the symbolic vector set.
Although the count of symbolic vectors remains consistent across all rows, the decoded version can still change. 
This is because capacity is not taken into account in the creation of symbolic vectors, and will only appear when decoding the set of symbolic vectors.

In addition, the table offers an approximate count of markings, calculated without the need to decode the structure.
Note that this number is related to the number of initial markings satisfying the CTL formula and not to the number of reachable markings from a given configuration.
Thus, we can efficiently handle markings until, potentially, infinity if the result converges. 
The earlier convergence occurs, the more efficient the computation becomes, enabling the capture of the Petri net's evolution at the earliest possible stage thanks to saturation.


\begin{table}[ht]
    \centering
    \begin{tabular}{|l|c|c|c|c||l|c|c|c|c|}
        \hline
        \# & Time\textbf{(s)} & SV nb & Marking nb & $\ast$ & \# & Time\textbf{(s)} & SV nb & Marking nb & $\ast$ \\
        \hline
        00 & 11 & 1 & $2.0 \times 10^{60}$ & F & 08 & 3.8 & 3 & $1.0 \times 10^{60}$ & F\\ 
        \rowcolor{lightgray}01 & 1.7 & 0 & 0 & F & 09 & 0.15 & 17 & $1.0 \times 10^{70}$ & T \\
        02 & 1.7 & 1 & $3.0 \times 10^{60}$ & F & 10 & 4 & 20 & $1.0 \times 10^{60}$ & F\\
        \rowcolor{lightgray}03 & 0.06 & 0 & 0 & F & 11 & 1350 & 1 & $1.0 \times 10^{84}$ & F\\
        04 & 20.5 & 17 & $1.0 \times 10^{70}$ & T & 12 & 0.28 & 6 & $1.1 \times 10^{70}$ & T\\
        \rowcolor{lightgray}05 & 2.4 & 1 & $1.0 \times 10^{70}$ & T &13 & 1.4 & 9 & $2.0 \times 10^{65}$ & T\\
        06 & 418 & 686 & $1.0 \times 10^{84}$ & T & 14 & 11.4 & 29 & $1.0 \times 10^{70}$ & T\\
        \rowcolor{lightgray}07 & 6.6 & 4 & $1.0 \times 10^{70}$ & F & 15 & 23.8 & 40 & $1.0 \times 10^{60}$ & F\\
        \hline
    \end{tabular}
    \caption{Computation of CTL fireability formulas for the \href{https://mcc.lip6.fr/2022/pdf/CircadianClock-form.pdf}{circadian clock} model with $N = 100~000$.
    Each CTL formula is numbered in $\#$ from $00$ to $15$, following the order provided in competition resources.
    $\ast$ column contains the local answer of each formula of the competition, where $F$ and $T$ stand for $false$ and $true$, respectively. 
    "nb" is the abbreviation for "number".\\}
    \label{table:benchmark:circadian-ctl-formula-100000}
\end{table}


In the context of the model checking contest, the same model is provided with varying initial markings, each time increasing the number of tokens.
The most challenging scenario for the circadian clock model is when \textit{N} is set to $100~000$, which means that several places in the Petri net contain such a number of tokens.

In the \textit{CTL fireability} category, three tools participated in 2022: \textit{GreatSPN}~\cite{babar:2010:greatspn}, \textit{ITS-Tools}~\cite{thierry:2015:symbolic}, and \textit{Tapaal}~\cite{jensen:2016:tapaal}. 
None of these tools managed to produce responses to all queries within the 60-minute time limit, underscoring the challenging nature of the task.
The best tool for this execution was Tapaal, which was able to answer 9 out of 16 queries\footnote{The results can be found \href{https://mcc.lip6.fr/2022/index.php?CONTENT=results/CTLFireability.html&TITLE=Results_for_CTLFireability}{here}, by looking for the \textit{CircadianClock} table.}.
In contrast, our results in~\Cref{table:benchmark:circadian-ctl-formula-100000} show that our symbolic technique successfully addressed all queries in approximately 30 minutes.
Furthermore, we computed the complete set of valid solutions and checked whether the initial marking of the model belongs to the resulting symbolic vector set. 
By focussing on a specific configuration, new optimisations become available, such as the \textit{on-the-fly} technique~\cite{bhat:1995:efficient}, \textit{structural reduction}~\cite{bonneland:2019:stubborn}, or \textit{stubborn reduction}~\cite{valmari:1991:stubborn}.
It should be noted that while our tool was not tested under identical conditions, the setup was not significantly different\footnote{Each execution on the virtual machine for the MCC was limited to a maximum of 16 GB, 2.4 GHz, and 4 cores.}. 

Nevertheless, our present implementation of canonical symbolic vector sets faces scalability challenges. Efficiency hinges on the implementation of canonical operations, whose complexity can easily grow exponentially with the number of places in the Petri net. Indeed, each additional place adds a component to the symbolic vectors, thereby increasing the complexity of the checks and operations required to preserve their canonical form. Our implementation strategy involves moving each element from one set to the other. Throughout this process, we must ensure that each moved element does not conflict with the elements in the other set. Additionally, when conflicts arise, the lexicographical order determines the symbolic vector that will contain the shared portion. Obtaining this shared part entails merging it with the lowest symbolic vector and subtracting it from the greatest symbolic vector. This subtraction may result in the creation of new symbolic vectors, requiring the canonical union to be reapplied to each of them and thus causing a deeper level of recursion.

This limitation also explains the selection of models used in our experiments. Since the proposed method directly addresses global model checking, the current implementation first computes a symbolic representation of the entire state space and only then checks whether the specific marking under consideration belongs to it. We therefore selected models from the Model Checking Contest with a relatively small number of places. The results presented here focus on the \textit{CircadianClock} model, while additional experiments were conducted on the \textit{ERK} and \textit{Robot Manipulation} models. The complete results of these experiments are available in~\cite{morard:2024:symbolic}.

In addition, most operations rely on canonical union.
Hence, we restricted our testing to a relatively small-sized Petri net due to these limitations.
Despite the inherent complexity, our method allows us to compute formulas that even the most proficient tools find challenging.
This approach already demonstrates its viability in specific scenarios, and we are confident that further refinement can address the mentioned issue and enhance its capabilities.

\section{Background \& Related works}

The original framework, initially introduced as \textit{Predicate structures} in~\cite{racloz:1994:properties}, remained stagnant for almost 30 years without further development. 
In its initial version, performance limitations rendered it nearly unusable or, at the very least, non-competitive.
Building on this foundation, we have expanded and enhanced it by introducing new operations, a novel canonical form, and optimisations.

In the realm of techniques addressing the state-space explosion problem in model checking, our method resembles symbolic model checking approaches.
Decision diagrams~\cite{akers:1978:binary,couvreur:2002:data,hamez:2008:hierarchical,hostettler:2011:high,tovchigrechko:2009:efficient} are the prevailing methods in symbolic model checking, employing various types of data encoding and optimisations~\cite{amparore:2018:decision,ciardo:2001:saturation,lopez:2014:stratagem,schwarick:2020:efficient}.
These representations enhance the efficiency of data representation and leverage homomorphisms for symbolic processing.
In contrast to our approach, based on encapsulation akin to intervals, decision diagrams exploit shared representations, allowing common components to be shared among different entities when encoding a set.
Moreover, our representation can handle certain infinite representations, a capability not universally shared by decision diagrams.
Among the big family of decision diagrams, Interval decision diagrams (IDDs)~\cite{schwarick:2020:efficient} appear to be the most closely related to our work.
However, symbolic vector sets generalise intervals for vectors, while IDDs consist of vectors of intervals.

\section{Conclusion}

In this paper, we introduce a novel structure termed \textit{symbolic vector set}. 
Alongside this structure, we present homomorphic operations, canonical form, and saturation optimisation to enhance computational efficiency.
Using this approach, we encode Petri net markings and employ the symbolic representation for CTL formula evaluation. 
Our initial results are promising, indicating potential avenues for further explorations and analysis.
We believe that our current implementation can be improved, particularly by enhancing the canonical operations.

Our work opens up several research directions.
Exploring our approach within the realm of local model checking holds promise.
Additionally, exploring models beyond Petri nets is another line of research, contingent upon our ability to construct the $pre$ operation symbolically for the model and to have the lemma \ref{lemma:homomorphicsvs} valid. In this case the monotonicity property play an important role, as its link to ordering indicates.
One of the most promising directions for our research is to integrate our method with other symbolic model checking methods, such as decision diagrams or symbolic encoding of the bounds themselves. 
We believe that both approaches are orthogonal and could mutually benefit from each other.

%
%
%
\bibliographystyle{splncs04}
\bibliography{references}

\end{document}

%% file: Package_and_commands/packages.tex
\usepackage{graphicx}

\usepackage{tabularray}

\usepackage{makecell}

\usepackage{tikz}
\usepackage{amsmath}
\usepackage{amssymb}
\usepackage{orcidlink}
\usepackage{stmaryrd}
\usepackage{subcaption}
\usepackage{bm}
\usepackage{listings}
\usetikzlibrary{patterns}

\lstdefinelanguage{Swift}{
    keywords={func, let, var, if, else, while, for, in, do, try, catch, throw, switch, case, default, return, break},
    sensitive=true,
    morecomment=[l]{//},
    morecomment=[s]{/*}{*/},
    morestring=[b]",
}

\lstdefinestyle{swiftCode}{
    language=Swift,
    basicstyle=\ttfamily\small,
    mathescape=true,  
    keywordstyle=\color{red},
    commentstyle=\color{gray!80},
    stringstyle=\color{red},
    backgroundcolor=\color{gray!8},
    frame=single,
    breaklines=true,
    numbers=left,
    numberstyle=\tiny\color{gray},
    stepnumber=1,
    tabsize=2,
    classoffset=1,
    morekeywords={sorted, collectLabelWeights, first, isNil},
    keywordstyle=\color{blue},
    classoffset=2,
}

\usepackage
[
    noabbrev,   
    nameinlink, 
]
{cleveref} 

\usepackage
[
    textsize = scriptsize, 
]
{todonotes}            

\usepackage{colortbl}
\usepackage{multirow}

%% file: Package_and_commands/commands.tex
\newcommand*{\powerset}[1]{\ensuremath{\mathcal{P}(#1)}}

\newcommand*{\tuple}[1]{\ensuremath{\langle #1 \rangle}}

\newcommand*{\set}[1]{\{#1\}}

\newcommand*{\SVS}{\ensuremath{\powerset{\mathbb{\SV}}}}
\newcommand*{\SV}{\ensuremath{\mathbb{SV}}}

\newcommand*{\Nat}{\ensuremath{\mathbb{N}}}
\newcommand*{\Bool}{\ensuremath{\mathbb{B}}}

\newcommand*{\true}{\ensuremath{true}}
\newcommand*{\false}{\ensuremath{false}}

\newcommand*{\ccap}{\ensuremath{\cap^{c}}}
\newcommand*{\ccup}{\ensuremath{\cup^{c}}}
\newcommand*{\csetminus}{\ensuremath{\setminus^{c}}}
\newcommand*{\cneg}{\ensuremath{\neg^{c}}}
\newcommand*{\usf}{\ensuremath{uf_{svs}}}
\newcommand*{\eval}[1]{\ensuremath{\llbracket #1 \rrbracket}} 

\usetikzlibrary{arrows.meta}
\usetikzlibrary{decorations.pathreplacing}

\usetikzlibrary{arrows,automata,calc,backgrounds,petri,positioning,shapes}
\tikzstyle{transition}=[draw, minimum width=0.6cm, minimum height=0.6cm]
\tikzstyle{every state}=[
    draw = black,
    thick,
    fill = white,
    minimum size = 4mm
]

\usetikzlibrary{arrows,shapes,snakes,automata,backgrounds,petri}
\usetikzlibrary{arrows.meta,decorations.pathmorphing,backgrounds,positioning,fit,petri}

 \tikzstyle{place}=[circle,thick,draw=blue!75,fill=blue!20,minimum size=6mm]
  \tikzstyle{red place}=[place,draw=red!75,fill=red!20]
  \tikzstyle{transition}=[rectangle,thick,draw=black!75,
  			  fill=black!20,minimum size=4mm]
    \tikzstyle{red transition}=[rectangle,thick,draw=red!75,
  			  fill=red!20,minimum size=4mm]
  \tikzstyle{every label}=[red]

\usetikzlibrary{automata, arrows.meta, positioning}

\tikzset{cross/.style={cross out, draw=black, minimum size=2*(#1-\pgflinewidth), inner sep=0pt, outer sep=0pt},
cross/.default={8pt}}

\makeatletter
\def\squarecorner#1{
    \pgf@x=\the\wd\pgfnodeparttextbox%
    \pgfmathsetlength\pgf@xc{\pgfkeysvalueof{/pgf/inner xsep}}%
    \advance\pgf@x by 2\pgf@xc%
    \pgfmathsetlength\pgf@xb{\pgfkeysvalueof{/pgf/minimum width}}%
    \ifdim\pgf@x<\pgf@xb%
        \pgf@x=\pgf@xb%
    \fi%
    \pgf@y=\ht\pgfnodeparttextbox%
    \advance\pgf@y by\dp\pgfnodeparttextbox%
    \pgfmathsetlength\pgf@yc{\pgfkeysvalueof{/pgf/inner ysep}}%
    \advance\pgf@y by 2\pgf@yc%
    \pgfmathsetlength\pgf@yb{\pgfkeysvalueof{/pgf/minimum height}}%
    \ifdim\pgf@y<\pgf@yb%
        \pgf@y=\pgf@yb%
    \fi%
    \ifdim\pgf@x<\pgf@y%
        \pgf@x=\pgf@y%
    \else
        \pgf@y=\pgf@x%
    \fi
    \pgf@x=#1.5\pgf@x%
    \advance\pgf@x by.5\wd\pgfnodeparttextbox%
    \pgfmathsetlength\pgf@xa{\pgfkeysvalueof{/pgf/outer xsep}}%
    \advance\pgf@x by#1\pgf@xa%
    \pgf@y=#1.5\pgf@y%
    \advance\pgf@y by-.5\dp\pgfnodeparttextbox%
    \advance\pgf@y by.5\ht\pgfnodeparttextbox%
    \pgfmathsetlength\pgf@ya{\pgfkeysvalueof{/pgf/outer ysep}}%
    \advance\pgf@y by#1\pgf@ya%
}
\makeatother

\pgfdeclareshape{square}{
    \savedanchor\northeast{\squarecorner{}}
    \savedanchor\southwest{\squarecorner{-}}

    \foreach \x in {east,west} \foreach \y in {north,mid,base,south} {
        \inheritanchor[from=rectangle]{\y\space\x}
    }
    \foreach \x in {east,west,north,mid,base,south,center,text} {
        \inheritanchor[from=rectangle]{\x}
    }
    \inheritanchorborder[from=rectangle]
    \inheritbackgroundpath[from=rectangle]
}

\definecolor{lightgray}{rgb}{0.9, 0.9, 0.9} 